\documentclass[3p, sort&compress]{elsarticle}
\usepackage{amsmath}
\usepackage{amsthm}
\usepackage{amssymb}
\usepackage{titlesec}
\usepackage[thinlines,thiklines]{easybmat}
\newcommand{\cent}{{|}\!\!\mathrm{c}}
\theoremstyle{definition}\newtheorem{definition}{Definition}
\theoremstyle{plain}\newtheorem{theorem}{Theorem}
\theoremstyle{definition}\newtheorem{remark}{Remark}
\theoremstyle{definition}\newtheorem{example}{Example}
\theoremstyle{plain}
\theoremstyle{plain}
\theoremstyle{plain}\newtheorem{lemma}{Lemma}
\theoremstyle{plain}

\begin{document}
\begin{frontmatter}
\title{Promise problems solved by quantum and classical finite automata}

\author{Shenggen Zheng$^{1}$}
\author{Lvzhou Li$^{1}$}
\author{Daowen Qiu$^{1,}$\corref{one}}
\author{Jozef Gruska$^{2}$}

 \cortext[one]{Corresponding author (D. Qiu). {\it E-mail addresses:} issqdw@mail.sysu.edu.cn (D. Qiu), zhengshenggen@gmail.com (S. Zheng), lilvzh@mail.sysu.edu.cn(L. Li), gruska@fi.muni.cz (J. Gruska).}

\address{
   $^1$Department of Computer Science, Sun Yat-sen University, Guangzhou 510006,  China\\
   $^{2}$Faculty of Informatics, Masaryk University, Brno 60200, Czech Republic\\
}

\begin{abstract}

 The concept of {\em promise problems} was introduced and started to be systematically explored by Even, Selman, Yacobi, Goldreich, and other scholars. It has been argued that promise
problems should be seen as partial {\em decision problems} and as such that they are more
fundamental than decision problems and formal languages that
used to be considered as the basic ones for complexity theory.
The main purpose of this paper is to explore the promise problems accepted by classical, quantum and also
semi-quantum finite automata. More specifically, we first introduce two acceptance modes of promise problems, {\em recognizability} and {\em solvability},  and explore their basic properties. Afterwards, we show several results concerning descriptional complexity on promise problems. In particular, we  prove: (1) there is a promise problem that can be recognized exactly by {\em  measure-once one-way quantum finite automata} (MO-1QFA), but no {\em deterministic finite automata} (DFA) can recognize it;  (2) there is a promise problem that can be solved with error probability $\epsilon\leq 1/3$ by {\em one-way finite automaton with quantum and classical states} (1QCFA), but no {\em one-way probability finite automaton} (PFA) can solve it with error probability  $\epsilon\leq 1/3$;
and especially, (3) there are promise problems $A(p)$ with prime $p$ that can be solved {\em with any error probability} by MO-1QFA with only two quantum basis states, but they can not be solved {\em exactly} by any MO-1QFA with two quantum basis states; in contrast, the minimal PFA solving $A(p)$ {\it with any error probability} (usually smaller than $1/2$) has $p$ states. Finally, we mention a number of problems related to promise for further study.

\end{abstract}


\end{frontmatter}

\section{Introduction}

Informally, a promise problem is the problem to decide whether an object or process has a property $P_1$ or $P_2$, provided it is promised (known) to have a property $P_3$.

The concept of a promise problem was introduced explicitly in \cite{ESY84} and it has been argued there that promise problems are actually more fundamental for the study of computational theory issues than decision problems or, more formally, formal language versions/encodings of the decision problems.

Such a view on the fundamental importance of promise problems has been even more emphasized in the survey paper \cite{Gh06}, where also the following basic version of the promise problems has been introduced.

\begin{definition}
A promise problem over an alphabet $\Sigma$ is a  pair $(A_{yes}, A_{no})$ of disjoint subsets of $\Sigma^*$. The union $A_{yes}\cup A_{no}$ is then called the {\em promise} and $A_{yes}$ as well $A_{no}$ are called promise's components.
\end{definition}

The goal is then   to decide  whether $x\in A_{yes}$ or  $x\in A_{no}$ for  a given string $x$ from the promise set.
In a special (trivial) case the promise is the the whole set $\Sigma^*$. However, in general it may be very nontrivial to decided whether an input string is  in a given promise set.

In spite of the fact that both papers  \cite{ESY84,Gh06} have brought interesting problems and outcomes,   the study of promise problems did not get a proper momentum yet.

On the other side, the results concerning several promise problems in quantum information processing have had very large impact. They demonstrated that using quantum phenomena and processes one can solve several interesting promise problems with much less quantum queries (to quantum black boxes) than in the case only classical tools and queries (to classical black boxes) are available. The initial development in this area was culminated by the result of Simon \cite{Sim97} that the
promise problem he introduced can be solved with the polynomial number of quantum and classical queries but not with polynomial number of classical queries only even if probabilistic tools are used. The second promise problem is the Hidden Subgroup Problem for non-commutative groups, which took very large attention, especially its special cases, for example  integer factorization, due to Shor \cite{Shor97}, and can be now seen as one of the most fundamental, and still open, problems.

Almost all papers so far, especially papers  \cite{ESY84,Gh06},  dealt with promise problems in the context of such high level complexity classes as {\bf P}, {\bf NP}, {\bf BPP}, {\bf SZK}  and so on.

In this paper we start to explore promise problems on another level, namely, using classical and quantum or even semiquantum finite automata to attack some promise problems working in various (especially two special) modes. The remainder of the paper is organized as follows.  In Section 2, we recall the definitions of classical and quantum finite automata that will be used in the paper, and define two acceptance modes of promise problems, {\em recognizability} and {\em solvability} of promise problems by automata. Then, in Section 3,  we deal with the closure and ordering properties of promise problems. Afterwards, in Section 4, lower and upper bounds are derived concerning the state complexity in a promise problem between the promise and its two components.

In particular, we study some promise problems in terms of classical and quantum finite automata in Section 5, and obtain the following results: that there is a promise problem that can be recognized exactly by {\em  measure-once one-way quantum finite automata} (MO-1QFA), but no {\em deterministic finite automata} (DFA) can recognize it (Theorem \ref{ThMOQ-DFA});  there is a promise problem that  can be solved with any error probability  by {\em one-way finite automaton with quantum and classical states} (1QCFA), but no {\em one-way probability finite automaton} (PFA) can solve it with error probability   $\epsilon\leq 1/3$ (Theorem \ref{ThQCFA-PFA}).

Especially, in Section 5 we prove a hierarchic result concerning QFA. More exactly, we show that  there are promise problems $A(p)$ with size $p$ that can be solved {\em with any error probability} by MO-1QFA with only two quantum basis states, but they can not be solved {\em exactly} by any MO-1QFA with two quantum basis states (Theorem \ref{ThMOQ-MOQ}), and in contrast, the minimal PFA solving $A(p)$ with any error probability (usually smaller than $1/2$) has $p$ states (Theorem \ref{ThMOQ-PFA}). However, we do not know whether there is an MO-1QFA with more than {\it two} quantum basis states being able to solve exactly this promise problems $A(p)$.

In addition, the above result may give rise to a hierarchic problem for the classes solved by MO-1QFA in terms of different quantum basis states. More precisely, let ${\cal C}(P)_{n}$  denote the class of promise problems solved exactly by an MO-1QFA with $n$ quantum basis states. Then, whether does ${\cal C}(P)_{m}\subset {\cal C}(P)_n$ hold for $m\leq n$?  Therefore, in Section 6 we mention a number of problems related for further study.

\section{Preliminaries}
We introduce in this section  some basic concepts    and  notations concerning classical and quantum finite automata. For more on  quantum information processing and (quantum and semi-quantum) finite automata we refer the reader to \cite{Qiu02,Qiu04,Qiu09,Qiu11, Qiu12BC,Qiu12,Qiu15,Gru99,Hop,Nie00}.

\subsection{Deterministic finite automata}
In this subsection we recall the definition of {\em  deterministic finite automata}  (DFA) and give the definition of so-called {\em promise version     deterministic finite automata}   (pvDFA).
\begin{definition}\label{dfa}
A deterministic finite automaton (DFA) ${\cal A}$ is specified by a 5-tuple

\begin{equation}
\mathcal{A}=(S,\Sigma,\delta,s_0, S_a),
\end{equation}
where:
\begin{itemize}
\item $S$ is a finite set of classical states;
\item $\Sigma$ is a finite set of input symbols;
\item $s_{0}\in S$ is the initial state of the automaton;
\item $S_a\subseteq S$ is a set of accepting states;
\item $\delta$ is a transition function:
\begin{equation}
\delta:S\times\Sigma \rightarrow S.
\end{equation}
\end{itemize}
\end{definition}

For any $w\in\Sigma^*$ and $\sigma\in \Sigma$, we define
\begin{equation}
    \widehat{\delta}(s,w\sigma)=\widehat{\delta}(  \widehat{\delta}(s,w), \sigma)
\end{equation}
and if $w$ is the empty string, then
 \begin{equation}
    \widehat{\delta}(s,w\sigma)=\delta(s,\sigma).
\end{equation}

To every DFA ${\cal A}=(S,\Sigma,\delta,s_0, S_a)$ we assign a language ${\cal L}({\cal A})$ defined as following
\begin{equation}
{\cal L}({\cal A})=\{w\mid  \widehat{\delta}(s_0,w)\in S_a, w\in\Sigma^*\}.
\end{equation}

\begin{definition}
A language $L$ over an alphabet $\Sigma$ is {\em recognized} by a DFA ${\cal A}$ if for every $w\in\Sigma^*$
\begin{itemize}
  \item $ w\in L$ if and only if $\widehat{\delta}(s_0,w)\in S_a$.
  \item  $ w\not\in L$ if and only if  $\widehat{\delta}(s_0,w)\not\in S_a$.
\end{itemize}
\end{definition}

It is well known that a language $L$ is  recognized by a DFA if and only if $L$ is  regular.
To every DFA ${\cal A}$ we assign also the (maximal) promise problem ${\cal P}({\cal A})$  defined as follows
\begin{equation}
{\cal P}({\cal A})=\left({\cal P}_{yes}({\cal A})=\{ w\mid  \widehat{\delta}(s_0,w)\in S_a, w\in\Sigma^*\}, {\cal P}_{no}({\cal A})=\Sigma^*\setminus {\cal P}_{yes}({\cal A})\right).
\end{equation}

\begin{definition}
A promise problem $A = (A_{yes}, A_{no})$ is {\em solved} by a DFA ${\cal A}$ if  for every $w\in A_{yes}\cup A_{no}\subseteq\Sigma^*$
\begin{itemize}
  \item $w\in A_{yes}$  implies that $\widehat{\delta}(s_0,w)\in S_a$.
  \item  $ w\in A_{no}$ implies that $\widehat{\delta}(s_0,w)\not\in S_a$.
\end{itemize}
\end{definition}

\begin{definition}
A {\em promise version     deterministic finite automaton}   (pvDFA) ${\cal A}$ is specified by a 6-tuple
\begin{equation}\label{df1}
{\cal A}=(S,\Sigma,\delta,s_0, S_a,S_r),
\end{equation}
where $S_a$ is a set of accepting states  and  $S_r$ is a set of  rejecting states, respectively, and $S$, $\Sigma$, $\delta$, $s_0$ are defined as in  Definition \ref{dfa}.
\end{definition}

A DFA can be see as a special pvDFA with $S_a\cup S_r=S$.   If a pvDFA ${\cal A}$ is such that $S_a\cup S_r=S$, then it is equivalent to a DFA.  In such a case, we say that  ${\cal A}$ is a DFA.
To every pvDFA we assign a promise problem ${\cal P}({\cal A})$  defined as following
\begin{equation}
{\cal P}({\cal A})=\left({\cal P}_{yes}({\cal A})=\{ w\mid  \widehat{\delta}(s_0,w)\in S_a, w\in\Sigma^*\}, {\cal P}_{no}({\cal A})=\{ w\mid  \widehat{\delta}(s_0,w)\in S_r, w\in\Sigma^*\}\right).
\end{equation}

\begin{definition}
A promise problem $A = (A_{yes}, A_{no})$ is {\em recognized} by a pvDFA ${\cal A}$ if for every $w\in\Sigma^*$
\begin{itemize}
  \item  $ w\in A_{yes}$ if and only if $\widehat{\delta}(s_0,w)\in S_a$.
  \item  $ w\in A_{no}$ if and only if $\widehat{\delta}(s_0,w)\in S_r$.
\end{itemize}
\end{definition}
\begin{definition}
A promise problem $A = (A_{yes}, A_{no})$ is {\em solved} by a pvDFA ${\cal A}$ if for every  $w\in A_{yes}\cup A_{no}$
\begin{itemize}
  \item   $w\in A_{yes}$ implies that $\widehat{\delta}(s_0,w)\in S_a$.
  \item  $w\in A_{no}$ implies that $\widehat{\delta}(s_0,w)\in S_r$.
\end{itemize}
\end{definition}

If a language $L$ is recognized by a DFA , then we can find efficiently the minimal DFA ${\cal A}$ such that ${\cal L}({\cal A})=L$. If a promise problem $A$ is recognized by a pvDFA, then we can also find efficiently the minimal pvDFA ${\cal A}$  such that ${\cal P}({\cal A})=A$. More about that will be  in  Section \ref{s-propeties}.

We will see that for pvDFA  recognizability and solvability modes can be seen as much different.

\subsection{Quantum and semi-quantum finite automata basic models and working modes}
Quantum finite automata were  introduced by  Kondacs and Watrous \cite{Kon97} and also by Moore and Crutchfields \cite{Moo97}.
It has been proved that one-way quantum finite automata (1QFA)  with unitary operations and projective measurements are less powerful than one-way  classical finite automata (1FA) \cite{Amb98,LiQiu09}. However, 1QFA can be more succinct in recognizing languages  or solving promise problems
\cite{Amb98,Amb09,AmYa11,BMP03,BMP06,BMP14,Ber05,Fre09,GQZ14b,Yak10,ZhgQiu112,Zhg13,Zhg14}.

\begin{definition}A measure-once quantum finite automaton (MO-1QFA) ${\cal M}$ is specified by a 5-tuple
 \begin{equation}
{\cal M}=(Q,\Sigma,\{U_{\sigma}\,|\, \sigma\in\Sigma' \},|{0}\rangle,Q_a)
\end{equation}
where:

\begin{itemize}
\item  $Q$ is a finite set of orthonormal quantum (basis) states, denoted as $\{|i\rangle\mid 0\leq i< |Q|\}$;

 \item $\Sigma$ is a finite alphabet of input symbols and
$\Sigma'=\Sigma\cup \{|\hspace{-1.5mm}c,\$\}$ (where $|\hspace{-1.5mm}c$ will be used as the left end-marker and $\$$ as the right end-marker);
\item $|0\rangle\in Q$ is the initial quantum state;

\item $Q_a \subseteq Q$ denotes the set of
accepting basis states;

\item $U_{\sigma}$'s  ($\sigma\in\Sigma'$) are unitary operators.
\end{itemize}

The quantum state space of this model will be the $|Q|$-dimensional Hilbert space denoted ${\cal H}_Q$.

Each quantum basis state $|i\rangle$  in ${\cal H}_Q$ can be represented by a column vector with the $(i+1)$th entry being $1$ and other entries being $0$.
With this notational convenience we can describe the above model as follows:
\begin{enumerate}
  \item  The initial state $|0\rangle$ is represented as $|q_0\rangle=(1,\overbrace{0,\cdots,0}^{|Q|-1})^\mathrm{T}$.
  \item  The accepting set $Q_a$ corresponds to the projective operator $P_{acc}=\sum_{|i\rangle\in Q_a}|i\rangle\langle i|$.
\end{enumerate}

The computation of an MO-1QFA ${\cal M}$ on an input string
$x=\sigma_{1}\sigma_{2}\cdots\sigma_{n}\in\Sigma^{*}$ goes as
follows:  ${\cal M}$ ``reads" the input string from the left end-marker to the right end-marker,  symbol by symbol, and  the unitary matrices $U_{|\hspace{-1mm}c}, U_{\sigma_1},U_{\sigma_2},\cdots,U_{\sigma_n},U_{\$}$ are applied, one by one, always on the current state, starting with  $|0\rangle$ as the initial state. Finally, the projective
measurement $\{P_{acc}, I-P_{acc}\}$ is performed on the final state, in order to accept or reject
the input. Therefore, for an input string  $w=\sigma_1\sigma_2\cdots \sigma_n$,  ${\cal M}$ has the accepting probability
\begin{equation}
Pr[{\cal M}\ \mbox{accepts}\  w] =\|P_{acc}U_{\$}U_{\sigma_n}\cdots U_{\sigma_2}U_{\sigma_1}U_{|\hspace{-1mm}c}|0\rangle\|^2
\end{equation}
and the rejecting probability
\begin{equation}
 Pr[{\cal M}\ \mbox{rejects}\  w]=1- Pr[{\cal M}\ \mbox{accepts}\  w].
\end{equation}
\end{definition}

\begin{definition}A promise version of a measure-once quantum finite automaton (pvMO-1QFA) ${\cal M}$ is specified by a 6-tuple
 \begin{equation}
{\cal M}=(Q,\Sigma,\{U_{\sigma}\,|\, \sigma\in\Sigma'  \},|{0}\rangle,Q_a,Q_r)
\end{equation}
where:
  $Q$, $\Sigma$, $\Sigma'$, $|0\rangle$, $Q_a$, $U_{\sigma}$ are  as defined in an MO-1QFA,  $Q_r \subseteq Q$ ($Q_r\cap Q_a=\emptyset$) denotes the set of
rejecting basis states. The  set $Q_r$ corresponds to the projective operator $P_{rej}=\sum_{|i\rangle\in Q_r}|i\rangle\langle i|$.

For an input string $w=\sigma_1\sigma_2\cdots \sigma_n$,  ${\cal M}$ has the accepting probability
\begin{equation}
Pr[{\cal M}\ \mbox{accepts}\  w] =\|P_{acc}U_{\$}U_{\sigma_n}\cdots U_{\sigma_2}U_{\sigma_1}U_{|\hspace{-1mm}c}|0\rangle\|^2
\end{equation}
and the rejecting probability
\begin{equation}
 Pr[{\cal M}\ \mbox{rejects}\  w]=\|P_{rej}U_{\$}U_{\sigma_n}\cdots U_{\sigma_2}U_{\sigma_1}U_{|\hspace{-1mm}c}|0\rangle\|^2.
\end{equation}
\end{definition}

Another interesting  (important) model of {\em two-way finite automata with quantum and classical states} (2QCFA)--was introduced by Ambainis and Watrous \cite{Amb02} and   explored in \cite{LF15,Yak10,ZhgQiu112,Zhg12,Zhg13a,Zhg13}.  If restricting the read-head in a 2QCFA to be {\it one-way}, then it is natural to get {\em one-way finite automata with quantum and classical states} (1QCFA). That is,
 1QCFA are one-way versions of 2QCFA, studied by Zheng and Qiu {\it et al} \cite{ZhgQiu112}. It is worth mentioning that more previously a different but more practical model called as {\it one-way quantum finite automata together with classical states} (1QFAC) was proposed and studied by Qiu { \it et al} \cite{Qiu15}.
Informally, a 1QCFA can be seen as a DFA which has an access to a quantum
 memory of a constant size (dimension), upon which the automaton performs
quantum transformations and projective measurements. Given a finite set of quantum basis states $Q$, we denote by $\mathcal{H}(Q)$
the Hilbert space spanned by $Q$. Let
$\mathcal{U}(\mathcal{H}(Q))$ and $\mathcal{O}(\mathcal{H}(Q))$
denote  the sets of unitary operators and projective measurements over $\mathcal{H}(Q)$, respectively.

\begin{definition}
A {\em one-way finite automaton with quantum and classical states} (1QCFA) ${\cal A}$ is specified by a 10-tuple
\begin{equation}
{\cal M}=(Q,S,\Sigma,\Theta,\Delta,\delta,|q_{0}\rangle,s_{0},S_{a},S_{r})
\end{equation}
where:
\begin{enumerate}
\item $Q$ is a finite set of orthonormal quantum  states,  a basis of a Hilbert space  ${\mathcal{H}}_Q$ spanned by states from $Q$.
\item $S$ is a finite set of classical states.

\item $\Sigma$ is a finite alphabet of input symbols and
$\Sigma'=\Sigma\cup \{\cent ,\$\}$, where $\cent$ will be used as the left end-marker and $\$$ as the right end-marker.

\item $|q_0\rangle\in Q$ is the initial quantum state.
\item $s_0$ is the initial classical state.
\item $S_{a}\subset S$ and $S_{r}\subset S$, where $S_{a}\cap S_{r}=\emptyset$, are  sets of
the classical accepting and rejecting states, respectively.

\item $\Theta$ is a quantum transition function
\begin{equation}
\Theta: S\setminus(S_{a}\cup S_{r})\times \Sigma'\to \mathcal{U}(\mathcal{H}(Q)),
\end{equation}
assigning to each pair $(s,\gamma)\in S\setminus(S_{a}\cup S_{r})\times \Sigma'$ a
 unitary transformation.

\item $\Delta$ is a mapping
\begin{equation}
\Delta:S\times \Sigma'\rightarrow
\mathcal{O}(\mathcal{H}(Q)),
\end{equation}
where each $\Delta(s,\gamma)$ corresponds to a projective measurement
(a projective measurement will be taken each time a unitary
   transformation is applied; if we do not need a measurement,
we denote that $\Delta(s,\gamma)=I$, and we assume the result of the measurement to be a fixed $c$).

\item $\delta$ is a special transition function of classical states.
Let the results set of the measurement be $\mathcal{C}=\{c_{1},c_{2},\dots$,
$c_{s}\}$, then
\begin{equation}
\delta:S\times \Sigma' \times \mathcal{C}\rightarrow
S,
\end{equation}
 where $\delta(s,\gamma)(c_{i})=s'$ means
that if a tape symbol $\gamma\in \Sigma'$ is being
 scanned and the projective measurement result is $c_{i}$, then the  state $s$ is changed to $s'$.
\end{enumerate}
\end{definition}

Given an input $w=\sigma_1\cdots\sigma_n$, the word on the tape will be seen as
$w=\cent w\$$ (for convenience, we denote $\sigma_0=\cent $ and $\sigma_{n+1}=\$$).
Now, we define the behavior of 1QCFA ${\cal M}$ on any input word $w$.
The computation starts in the classical state $s_0$ and the quantum state $|q_0\rangle$. After that
transformations associated with symbols in the word  $\sigma_0\sigma_1\cdots, \sigma_{n+1}$ are applied in succession.
A transformation associated
with a state $s\in S$ and a symbol $\sigma\in\Sigma'$ consists of three steps:
\begin{enumerate}
\item  The unitary transformation $\Theta(s,\sigma)$ is applied to the current quantum state
$|\phi\rangle$, yielding the new state
$|\phi'\rangle=\Theta(s,\sigma)|\phi\rangle$.

\item The observable $\Delta(s,\sigma)=\mathcal{O}$ is measured on
$|\phi'\rangle$. The set of possible results is $\mathcal{C}=\{c_1, \cdots, c_s\}$.
According to   quantum mechanics principles, such a
measurement yields the classical outcome $c_k$ with probability
$p_k=||P(c_k)|\phi'\rangle||^2$, and the quantum state of ${\cal M}$
collapses to $P(c_k)|\phi'\rangle/\sqrt{p_k}$.

\item The current classical state $s$ is changed to $\delta(s,\sigma)(c_k) =s'.$

\end{enumerate}
An input word $w$ is assumed to be accepted (rejected) if and only if the automaton enters at the end
 an accepting (rejecting) state.  It is  assumed that $\delta$ is well defined so that 1QCFA ${\cal M}$ always accepts or
rejects at the end of the computation.

\begin{definition}
 An MO-1QFA, 1QCFA ${\cal M}$ recognizes a language $L$ with bounded error  $\varepsilon$ if for every $w\in\Sigma^*$
\begin{itemize}
\item  $w\in L$ if and only if $Pr[{\cal M}\  \text{accepts}\  w]\geq 1-\varepsilon$.
\item  $ w\notin L$ if and only if $Pr[{\cal M}\ \text{rejects}\  w]\geq 1-\varepsilon$.
\end{itemize}
\end{definition}

\begin{definition}
 A  pvMO-1QFA ${\cal M}$ recognizes a promise problem $A=(A_{yes},A_{no})$ with an error probability $\varepsilon$ if for every $w\in\Sigma^*$
\begin{itemize}
\item  $ w\in A_{yes}$ if and only if $Pr[{\cal M}\  \text{accepts}\  w]\geq 1-\varepsilon$.
\item  $w\in A_{no}$ if and only if $Pr[{\cal M}\ \text{rejects}\  w]\geq 1-\varepsilon$.
\end{itemize}
\end{definition}
\begin{definition}
 A promise problem $A = (A_{yes}, A_{no})$ is solved  by  a pvMO-1QFA  ${\cal M}$  with an error probability $\varepsilon$ if for every  $w\in A_{yes}\cup A_{no}$
\begin{itemize}
\item $ w\in A_{yes}$ implies that $Pr[{\cal M}\  \text{accepts}\  w]\geq 1-\varepsilon$, and
\item $ w\in  A_{no}$ implies that $Pr[{\cal M}\ \text{rejects}\  w]\geq 1-\varepsilon$.
\end{itemize}
\end{definition}
If $\varepsilon=0$, we say that the automaton ${\cal M}$ solves (recognizes) the promise problem $A$ exactly.

\section{Properties of pvDFA}\label{s-propeties}

We will now study closure properties
of promise problems recognized or solved by pvDFA.

\begin{theorem}\label{Th-rec}
A promise problem $A = (A_{yes}, A_{no})$ can be recognized by a pvDFA ${\cal A}$ iff both $A_{yes}$ and $A_{no}$ are regular.
\end{theorem}

\begin{proof}
($\Rightarrow$)
Suppose that a  promise problem $A$ can be recognized by a pvDFA ${\cal A}=(S,\Sigma,\delta,s_0, S_a,S_r)$. In such a case,    for all $w\in\Sigma^*$, $w\in A_{yes}$ if and only if $\widehat{\delta}(s_0,w)\in S_a$.
Let DFA ${\cal A}_{y}= (S,\Sigma,\delta,s_0, S_a)$. Obviously, $A_{yes}$ is recognized by ${\cal A}_{y}$ and therefore $A_{yes}$ is regular. Using similar argument, one can show that $A_{no}$ is regular.

($\Leftarrow$) Let us assume that the set $A_{yes}$  can be recognized by a DFA ${\cal A}^1=(S^1,\Sigma,\delta^1,s_{0}^1, S_{a}^1)$ and $A_{no}$  can be recognized by a DFA ${\cal A}^2=(S^2,\Sigma,\delta^2,s_{0}^2, S_{a}^2)$. We now consider the following pvDFA ${\cal A}=(S,\Sigma,\delta,s_{0}, S_{a},S_r)$ where
\begin{itemize}
\item $S=(S_1\times S_2)\setminus (S_{a}^1\times S_{a}^2)$;
\item $s_{0}=\langle s_{0}^1,  s_{0}^2\rangle$;
\item $\delta(\langle s^1,s^2\rangle, \sigma)=\langle \delta^1(s^1,\sigma),  \delta^2(s^2,\sigma)\rangle$;
\item $S_a=S_{a}^1\times (S^2\setminus S_{a}^2 )$ and $S_r=(S^1\setminus S_{a}^1 ) \times  S_{a}^2$.
\end{itemize}

For any $w\in \Sigma^*$,  we prove first that  $s=\widehat{\delta}(s_0,w)\not\in S_{a}^1\times S_{a}^2$.
Let us assume that $s=\langle s^1,  s^2\rangle\in S_{a}^1\times S_{a}^2$. We have $\widehat{\delta}(s_0,w)=\langle \widehat{\delta^1}(s_0^1,w), \widehat{\delta^2}(s_0^2,w)\rangle=\langle s^1,  s^2\rangle$. Therefore,  $\widehat{\delta^1}(s_0^1,w)=s_1\in  S_{a}^1$ and $\widehat{\delta^2}(s_0^2,w)=s_2\in  S_{a}^2$. This implies that  $w\in A_{yes}$ and $w\in A_{no}$, which is a contradiction.

If $w\in A_{yes}$, then $s^1=\widehat{\delta^1}(s_0^1,w)\in S_{a}^1$ and $s^2=\widehat{\delta^2}(s_0^2,w)\not\in S_{a}^2$. Therefore,  $\widehat{\delta}(s_0,w)=\langle\widehat{\delta^1}(s_0^1,w),\widehat{\delta^2}(s_0^2,w)\rangle=\langle s^1,  s^2\rangle\in S_{a}^1\times (S^2\setminus S_{a}^2 )=S_a$.

If $w\in \Sigma^*$ is such that $\widehat{\delta}(s_0,w)\in S_a$, then $\widehat{\delta}(s_0,w)=\langle\widehat{\delta^1}(s_0^1,w),\widehat{\delta^2}(s_0^2,w)\rangle=\langle s^1,  s^2\rangle\in S_{a}^1\times (S^2\setminus S_{a}^2 )$. We have therefore $\widehat{\delta^1}(s_0^1,w)\in S_{a}^1$ and  $w\in A_{yes}$.

With a similar argument as above, we can show that  for any $w\in \Sigma^*$,  $w\in A_{no}$ if and only if $\widehat{\delta}(s_0,w)\in S_r$.

Therefore the  promise problem $A = (A_{yes}, A_{no})$ can be recognized by the pvDFA ${\cal A}$.
\end{proof}

\begin{remark}
If a promise problem $A$ is recognized by a pvDFA ${\cal A}$, then  $A$ is solved by the same pvDFA ${\cal A}$. However, if a promise problem $A$ is solved by a pvDFA ${\cal A}$, it does not necessarily mean that  $A$ can be recognized by a pvDFA. For example, let us  consider  the  promise problems  $B^{l}=(B_{yes}^{l},B_{no}^{l})$ with $B_{yes}^{l}=\{a^ib^{i}\mid i\geq 0\}$ and $B_{no}^{l}=\{a^ib^{i+l}\mid i\geq 0\}$, where $l$ is a fix positive integer. The promise problem   $B^{l}$ can be solved   by a DFA \cite{GQZ14b}.  Therefore it  can be solved by a pvDFA.
  However, both $B_{yes}^{l}$ and $B_{no}^{l}$ are {\em nonregular languages}. Therefore  $B^{l}$ cannot be recognized by a pvDFA.
\end{remark}

\subsection{Pumping Lemmas}

The pumping lemma for pvDFA concerning recognition is similar to the classical one \cite{Hop}.
\begin{lemma}[{\bf Pumping Lemma I}]
Let a  promise problem $A = (A_{yes}, A_{no})$ can be recognized by a  pvDFA ${\cal A}$. Then there exists an integer $p \geq 1$, depending only on  ${\cal A}$, such that every string $w$ in $A_{yes}$ ($A_{no}$),  of length at least $p$, can be written as $w = xyz$ (i.e., w can be divided into three substrings), satisfying the following conditions:
\begin{itemize}
\item $|y| \geq 1$;
\item $|xy| \leq p$;
\item  $xy^tz\in A_{yes}\ (A_{no})$ for all integers $t\geq 0$.
\end{itemize}
\end{lemma}

The  pumping
lemma for pvDFA concerning solvability has quite a different form than the above Pumping Lemma.

\begin{lemma}[{\bf Pumping Lemma II}]\label{pumpingLemmaI}
Let a promise problem $A = (A_{yes}, A_{no})$ can be solved by  a pvDFA ${\cal A}$. Then there exists an integer $p\geq 1$, depending only on  ${\cal A}$, such that every string $w$ in $A_{yes}$ ($A_{no}$),  of length at least $p$, can be written as $w = xyz$ (i.e., w can be divided into three substrings), satisfying the following conditions:
\begin{itemize}
\item $|y| \geq 1$;
\item $|xy| \leq p$;
\item $xy^tz\notin A_{no}\ (A_{yes})$ for all integers $t\geq 0$.
\end{itemize}
\end{lemma}
\begin{proof}
Let  pvDFA  ${\cal A}=(S,\Sigma,\delta,s_0, S_a,S_r)$ and $p=|S|$ be the number of the of states of ${\cal A}$. For a word $w=\sigma_1\ldots\sigma_n\in A_{yes}\ (A_{no})$,  we denote the computation of  ${\cal A}$ on $w$ by the following sequence of transitions:
\begin{equation}
  s_0\xrightarrow{\sigma_1} s_1  \xrightarrow{\sigma_2} \cdots  \xrightarrow{\sigma_n}s_n,
\end{equation}
where $s_n\in S_a\ (S_r)$.

If $n\geq p$, then there exist $i< j$ such that $s_i=s_j$. Let $x=\sigma_1\ldots\sigma_i$, $y=\sigma_{i+1} \ldots \sigma_j$ and $z=\sigma_{j+1}\ldots\sigma_n$.  We have, $\widehat{\delta}(s_0,x)=s_i$,  $\widehat{\delta}(s_i,y)=s_j$ and $\widehat{\delta}(s_j,z)=s_n\in S_a$. Therefore $\widehat{\delta}(s_i,y^*)=s_i$.

  If there exists an integer $t\geq 0$  such that $w=xy^tz\in A_{no}\ (A_{yes})$, then $\widehat{\delta}(s_0,w)=\widehat{\delta}(s_0,xy^tz)=\widehat{\delta}(s_i,y^tz)=\widehat{\delta}(s_i,z)=s_n\in S_a\ (S_r)$, which is a contradiction. Therefore, we have $xy^tz\notin A_{no}\ (A_{yes})$ for all $t\geq 0$.
\end{proof}

In the following example it will be shown how  Pumping Lemma II can be used to prove that a promise problem can not be solved by pvDFA.
\begin{example}\label{eample-1}
 Let us consider the promise problem $C=({C_{yes},C_{no}})$ with $C_{yes}=\{a^nb^n\}$ and $C_{no}=\{a^nb^m\,|\, n\neq m\}$.
 Assume that $C$ can  be solved by a pvDFA ${\cal A}$ and $p$ is the constant for the pumping lemma.
Choose $w=a^pb^p\in A_{yes}$. Clearly, $|w|>p$. By the Pumping Lemma II, $w=xyz$ for some $x,y,z\in\Sigma^*$ such that (1) $|xy|\leq p$, (2) $|y|\geq 1$, and (3) $xy^tz\not\in A_{no}$ for all $t\geq 0$. By (1) and (2), we have $y=a^{k}$, $1\leq k\leq p$.   However, $xy^2z=a^{p+k}b^p\in A_{no}$. Therefore, (3) does not hold. The promise problem $C$ therefore does not satisfy the pumping property of the Pumping Lemma II.
Hence,  the promise problem $C$ can not be solved by any pvDFA.

 \end{example}

\subsection{Closure properties}

Let us have promise problems $A = (A_{yes}, A_{no})$ and  $B = (B_{yes}, B_{no})$ over the same  alphabet\footnote{When we take the union or intersection of two promise problems, they might have different alphabets. However, if  $P = (P_{yes}, P_{no})$ is a promise problem over  alphabet $\Sigma$, then we can also think of $P$ over any finite alphabet that is a superset of $\Sigma$.  See \cite{Hop} for more details.}. The complement, intersection and union operations on such promise  problems will be defined as follows.

\begin{itemize}

  \item Complement: $\overline{A}=(\overline{A}_{yes}, \overline{A}_{no})$, where $\overline{A}_{yes}=A_{no}$ and $\overline{A}_{no}=A_{yes}$.

\item  Intersection: $C=A\cap B=(C_{yes}, C_{no})$, where $C_{yes}=A_{yes}\cap B_{yes}$ and $C_{no}=A_{no}\cap B_{no}$.

  \item Union: if $(A_{yes}\cup B_{yes})\cap (A_{no}\cup B_{no})\neq \emptyset$, then the union of $A$ and $B$ will be undefined; otherwise the union  $C=A\cup B=(C_{yes}, C_{no})$, where $C_{yes}=A_{yes}\cup B_{yes}$ and $C_{no}=A_{no}\cup B_{no}$.
\end{itemize}
There seems to be several
other ways one could  try to define intersection and union of $A$ and $B$.  We will now try to argue  that our definitions are reasonable. Let us assume that  Alice has two subsets $A_{yes}$ and $A_{no}$ over $\Sigma^*$.   If  Alice would be asked for an $x\in  A_{yes}\cup A_{no}$ whether  $x\in A_{yes}$ or $x\in A_{no}$, then she should be able to  answer  ``yes" or ``no" (by checking whether $x\in A_{yes}$ or $x\in A_{no}$). Let us assume also that  Bob has two subsets $B_{yes}$ and $B_{no}$ over $\Sigma^*$. If  Bob  would be asked for an $x\in  B_{yes}\cup B_{no}$ whether  $x\in B_{yes}$ or $x\in B_{no}$,   then he should be able to give  correct answer.
The intersection of two promise problems should be therefore such that for a given input,   both Alice and Bob are able to tell  whether a given input is in the yes--set or no--set.

 The union  of two promise problems should be therefore such that for a given input,   at least one of  Alice and Bob are able to tell whether it is in the yes--set or no--set, that is why union was defined in the way it was.

Let us now give several results concerning how promise
problems are closed on some operations in the case of recognizability and solvability modes.

\begin{theorem}
If a promise problems $A$ can be recognized (solved) by a pvDFA, then  $\overline{A}$ can be recognized (solved) by a pvDFA.
\end{theorem}
\begin{proof}
Suppose that a  promise problem $A$ can be recognized (solved) by a pvDFA ${\cal A}=(S,\Sigma,\delta,s_0, S_a,S_r)$. Exchanging the sets of accepting states and rejecting states of the pvDFA ${\cal A}$, we get a new   pvDFA ${\cal A}'=(S,\Sigma,\delta,s_0, S_r,S_a)$. It is easy to see that $\overline{A}$ is recognized (solved) by the pvDFA ${\cal A}'$.
\end{proof}

\begin{theorem}
If promise problems $A$ and $B$ can be recognized by pvDFA, then their intersection can be also recognized by a pvDFA.
\end{theorem}
\begin{proof}
Suppose that a  promise problem $A$ can be recognized by  a pvDFA ${\cal A}^1=(S^1,\Sigma,\delta^1,s_0^1, S_a^1,S_r^2)$ and a  promise problem $B$ can be recognized by a pvDFA ${\cal A}^2=(S^2,\Sigma,\delta^2,s_0^2, S_a^2,S_r^2)$.
We consider a pvDFA ${\cal A}=(S,\Sigma,\delta,s_0, S_a,S_r)$, where

\begin{itemize}
\item $S=S^1\times S^2$;
\item $s_{0}=\langle s_{0}^1,  s_{0}^2\rangle$;
\item $\delta(\langle s^1,s^2\rangle, \sigma)=\langle \delta^1(s^1,\sigma),  \delta^2(s^2,\sigma)\rangle$;
\item $S_a=S_{a}^1\times S_a^2 $ and $S_r=S_r^1\times  S_r^2$.
\end{itemize}

Let the promise problem $C=(C_{yes}, C_{no})$ be the intersection of the promise problems $A = (A_{yes}, A_{no})$ and  $B = (B_{yes}, B_{no})$.

If $w\in C_{yes}$, then $w\in A_{yes}\cap B_{yes}$. We have $\widehat{\delta^1}(s_0^1,w)\in S_a^1$ and  $\widehat{\delta^2}(s_0^2,w)\in S_a^2$. Therefore, we have
$\widehat{\delta}(s_0,w)=\widehat{\delta}(\langle s_0^1,s_0^2  \rangle,w)=\langle \widehat{\delta^1}(s_0^1,w), \widehat{\delta^2}(s_0^2,w)  \rangle\in S_a^1\times S_a^2=S_a$.

If $w\in \Sigma^*$ is such that $\widehat{\delta}(s_0,w)\in S_a$, we have $\widehat{\delta}(s_0,w)=\widehat{\delta}(\langle s_0^1,s_0^2  \rangle,w)=\langle \widehat{\delta^1}(s_0^1,w), \widehat{\delta^2}(s_0^2,w)  \rangle\in S_a= S_a^1\times S_a^2$. Therefore,  $\widehat{\delta^1}(s_0^1,w)\in S_a^1 $ and $\widehat{\delta^2}(s_0^2,w)\in S_a^2$, i.e. $w\in A_{yes}$ and $w\in B_{yes}$. Hence, $w\in  A_{yes}\cap B_{yes}=C_{yes}$.

Therefore, we have $w\in C_{yes}$ if and only if $\widehat{\delta}(s_0,w)\in S_a$. By a similar argument, we can show that   $w\in C_{no}$ if and only if $\widehat{\delta}(s_0,w)\in S_r$.
Hence, the promise problem $C=A\cap B$ can  be recognized by the pvDFA ${\cal A}$.
\end{proof}

\begin{theorem}
If promise problems $A$ and $B$ can be solved  by pvDFA, then their intersection can be solved also by a pvDFA.
\end{theorem}
\begin{proof}
Let a promise problem $C=(C_{yes}, C_{no})$ be the intersection of  the two promise problems $A = (A_{yes}, A_{no})$ and  $B = (B_{yes}, B_{no})$.
Suppose that the promise problem $A = (A_{yes}, A_{no})$ can be solved by a pvDFA ${\cal A}$.
Since $C_{yes}=A_{yes}\cap B_{yes}\subset A_{yes}$ and $C_{no}=A_{no}\cap B_{no}\subset A_{no}$,  the  promise problem $C$ can be solved by ${\cal A}$.
\end{proof}

\begin{theorem}
Let promise problems A and B over an alphabet $\Sigma$ can be recognized by pvDFA and their union $C$ exists, then $C$ can be recognized also by a pvDFA.
\end{theorem}
\begin{proof}
Suppose that the  promise problem $A$ with the alphabet $\Sigma$ can be recognized by a pvDFA ${\cal A}^1=(S^1,\Sigma,\delta^1,s_0^1, S_a^1,S_r^2)$ and the  promise problem $B$ with alphabet $\Sigma$ can be recognized by  a pvDFA ${\cal A}^2=(S^2,\Sigma,\delta^2,s_0^2, S_a^2,S_r^2)$.

We consider the pvDFA ${\cal A}=(S,\Sigma,\delta,s_0, S_a,S_r)$, where

\begin{itemize}
\item $S=(S_1\times S_2)\setminus ((S_{a}^1\times S_{r}^2)\cup (S_{r}^1\times S_{a}^2))$;
\item $s_{0}=\langle s_{0}^1,  s_{0}^2\rangle$;
\item $\delta(\langle s^1,s^2\rangle, \sigma)=\langle \delta^1(s^1,\sigma),  \delta^2(s^2,\sigma)\rangle$;
\item $S_a=\{\langle s^1,s^2 \rangle \mid s^1\in S_a^1 \ \mbox{or}\ s^2\in S_a^2\}$ and $S_r=\{\langle s^1,s^2 \rangle \mid s^1\in S_r^1 \ \mbox{or}\ s^2\in S_r^2\}$.
\end{itemize}

Let the promise problem $C=(C_{yes}, C_{no})$ be the union of  promise problems $A = (A_{yes}, A_{no})$ and  $B = (B_{yes}, B_{no})$.
Since the union $C=A\cup B$ exists,  we have $(A_{yes}\cup B_{yes})\cap (A_{no}\cup B_{no})= \emptyset$.

We prove now for any $w\in \Sigma^*$ that  $s=\widehat{\delta}(s_0,w)\not\in S_{a}^1\times S_{r}^2$.
Let us assume that $s=\langle s^1,  s^2\rangle\in S_{a}^1\times S_{r}^2$.
We have $\widehat{\delta}(s_0,w)=\langle \widehat{\delta^1}(s_0^1,w), \widehat{\delta^2}(s_0^2,w)\rangle=\langle s^1,  s^2\rangle$.
Therefore,  $\widehat{\delta^1}(s_0^1,w)=s_1\in  S_{a}^1$ and $\widehat{\delta^2}(s_0^2,w)=s_2\in  S_{r}^2$. From that it follows that  $w\in A_{yes}$ and $w\in B_{no}$.
Therefore, $w\in(A_{yes}\cup B_{yes})\cap (A_{no}\cup B_{no})=\emptyset$,
 which is a contradiction. By a similar argument we can prove that $\widehat{\delta}(s_0,w)\not\in S_{r}^1\times S_{a}^2$. Hence $S_a\cap S_r=\emptyset$.

If $w\in C_{yes}$, then $w\in A_{yes}\cup B_{yes}$. We have $\widehat{\delta^1}(s_0^1,w)\in S_a^1$ or  $\widehat{\delta^2}(s_0^2,w)\in S_a^2$. Therefore,
$\widehat{\delta}(s_0,w)=\widehat{\delta}(\langle s_0^1,s_0^2  \rangle,w)=\langle \widehat{\delta^1}(s_0^1,w), \widehat{\delta^2}(s_0^2,w)  \rangle\in S_a$.

If $w\in \Sigma^*$ is such that $\widehat{\delta}(s_0,w)\in S_a$, we have $\widehat{\delta}(s_0,w)=\widehat{\delta}(\langle s_0^1,s_0^2  \rangle,w)=\langle \widehat{\delta^1}(s_0^1,w), \widehat{\delta^2}(s_0^2,w)  \rangle\in S_a$. Therefore,  $\widehat{\delta^1}(s_0^1,w)\in S_a^1 $ and $\widehat{\delta^2}(s_0^2,w)\in S_a^2$, i.e. $w\in A_{yes}$ or $w\in B_{yes}$. Hence, $w\in  A_{yes}\cup B_{yes}=C_{yes}$.

Therefore,  $w\in C_{yes}$ if and only if $\widehat{\delta}(s_0,w)\in S_a$. By a similar argument, we can show that   $w\in C_{no}$ if and only if $\widehat{\delta}(s_0,w)\in S_r$.
Hence, the promise problem $C=A\cup B$ can  be recognized by the pvDFA ${\cal A}$.
\end{proof}

\begin{remark}
If promise problems $A$ and $B$ can be solved by pvDFA and their union $C$ exists, then $C$ may not be solved by a pvDFA. Indeed,
let $A = (A_{yes}, A_{no})$, where $A_{yes}=\{a^nb^n \,|\, n\mbox{ is odd}\}$ and  $A_{no}=\{a^nb^m \,|\, m\neq n\mbox{ and at least one of } m,\linebreak[0]n \mbox{ is even}\}$.
If $w\in A_{yes}$, then $\#_a(w)$ and $\#_b(w)$ are odd. If $w\in A_{no}$, at least one of $\#_a(w)$ and $\#_b(w)$ is even.  Obviously, we can design a pvDFA to solve the promise problem $A$.
 Let $B = (B_{yes}, B_{no})$, where $B_{yes}=\{a^nb^n \,|\, n\mbox{ is even}\}$ and  $B_{no}=\{a^nb^m \,|\, m\neq n\mbox{ and at least one } \linebreak[0]\mbox{of } m,n \mbox{ is odd}\}$.  Similarly,  we can design another pvDFA to solve the promise problem $B$.
 Now we consider their union $C=A\cup B=({C_{yes},C_{no}})$, where $C_{yes}=A_{yes}\cup B_{yes}=\{a^nb^n\}$ and $C_{no}=A_{no}\cup B_{no}=\{a^nb^m\,|\, n\neq m\}$. According to Example \ref{eample-1},
 $C$ can not be solved by any pvDFA.

\end{remark}

\subsection{Ordering}
Let us start with some basic definitions concerning ordering of promise problems.
Let $A=(A_{yes},A_{no})$ and $B=(B_{yes},B_{no})$ be two promise problems over an alphabet $\Sigma$.
We say that $A$ is a {\em subproblem} of  $B$, denoted by $A\leq B$, if $A_{yes}\subseteq B_{yes}$ and $A_{no}\subseteq B_{no}$.

We say  also that a pvDFA ${\cal A}$ is {\em equivalent} to a pvDFA ${\cal B}$ (denoted by ${\cal A}$=${\cal B}$) if ${\cal P}({\cal A})={\cal P}({\cal B})$.
 We say that a pvDFA ${\cal B}$  is {\em more powerful than or equivalent to} a pvDFA ${\cal A}$ (denoted by ${\cal B}\geq {\cal A}$ or ${\cal A}\leq {\cal B}$ ) if ${\cal P}({\cal A})\leq {\cal P}({\cal B})$.
It is clear that the set of all
pvDFA is a partially ordered set with the partial order `$\leq$'.
  We say that a pvDFA ${\cal B}$  is {\em more powerful} than  a pvDFA ${\cal A}$ (denoted by ${\cal B}> {\cal A}$ or ${\cal A}< {\cal B}$ ) if ${\cal P}({\cal A})\leq {\cal P}({\cal B})$
  and ${\cal P}({\cal A})\neq {\cal P}({\cal B})$.

The first outcome concering the impact of ordering on solvability of promise problems follows in  a straightforward way from basic definitions.
\begin{theorem}
 If a promise problem $A$ can be solved by a pvDFA ${\cal A}$ and ${\cal A}\leq {\cal B}$, then the promise problem $A$ can be solved by the pvDFA ${\cal B}$.
 \end{theorem}

 We say a pvDFA ${\cal A}$ is maximally powerful if there does not exist a pvDFA ${\cal B}$ such that ${\cal A}<{\cal B}$.

\begin{theorem}
 A pvDFA ${\cal A}$ is  maximally  powerful if and only if it is (essentially) a DFA.
\end{theorem}
\begin{proof}

If  ${\cal A}$ is a DFA, then ${\cal P}_{yes}({\cal A})=\Sigma^*\setminus {\cal P}_{no}({\cal A})$.
Therefore, there does not exist a promise problem $B$ such that ${\cal P}({\cal A}) <B$.  Therefore, there does exist a pvDFA ${\cal B}$ such that ${\cal A}<{\cal B}$, i.e.  ${\cal A}$ is maximally powerful.

Assume that  a pvDFA ${\cal A}$ is maximally powerful and ${\cal A}$ is not a DFA.
Suppose that the pvDFA  ${\cal A}=(S,\Sigma,\delta,s_0, S_a, S_r)$ and it is state minimal.
 We have that $S_a\cup S_r\neq S$ and $S_r$ is a proper subset of $S\setminus S_a$.
Let us now consider a new pvDFA ${\cal B}=(S,\Sigma,\delta,s_0, S_a,S\setminus S_a)$.
Suppose that ${\cal P}({\cal B})=(B_{yes},B_{no})$.
Therefore, there must exist some  $w\in B_{no}$ such that $\widehat{\delta}(s_0,w)\in S\setminus S_a$ and  $\widehat{\delta}(s_0,w)\not\in S_r$. Therefore, ${\cal P}_{no}({\cal A})$ is a proper subset of $B_{no}$. Since ${\cal P}_{yes}({\cal A})=B_{yes}$. We have ${\cal P}({\cal A})<{\cal P}({\cal B})$, which is a contradiction. Hence, ${\cal A}$ must be a DFA.
\end{proof}

%From the above two theorems, we can see that  if pvDFA ${\cal A}\leq {\cal B}$,   then we can use  ${\cal B}$ to replace  ${\cal A}$ in solving promise problems without changing   ${\cal B}$.
%Since DFA are maximally powerful, DFA ${\cal A}\leq {\cal B}$ will mean that ${\cal A}= {\cal B}$. Therefore,   we can  use DFA ${\cal B}$ to replace DFA ${\cal A}$  in solving promise problems  unless ${\cal B}={\cal A}$.

 We say that two pvDFA ${\cal A}$ and  ${\cal B}$ are comparable if ${\cal A}={\cal B}$ or ${\cal A}< {\cal B}$ or ${\cal A}> {\cal B}$.
  Two DFA are either equivalent or not comparable.  If a pvDFA ${\cal A}$  is a DFA, then there does not exist a pvDFA ${\cal B}$ such that ${\cal A}<{\cal B}$.
 Equivalence of two DFA can be seen as a special case of the equivalence of two pvDFA.

If pvDFA ${\cal A}={\cal B}$, then  ${\cal A}$ is a potential substitute for ${\cal B}$ in recognizing promise problems (languages).
If pvDFA ${\cal A}\geq {\cal B}$, then  ${\cal A}$ is a potential substitute for ${\cal B}$ in solving promise problems.   Therefore, it is important to determine the order of pvDFA.

In order to study determination of equivalence and ordering of two given pvDFA,  we now introduce
the concept of a {\em  bilinear machine} (BLM).

By \cite{LQ08}, a
BLM over an alphabet $\Sigma$ is a four-tuple ${\cal A}=( S,
\pi,\{M(\sigma)\}_{\sigma\in\Sigma},\eta)$, where $S$ is a finite set of states
 with $|S|=n$, $\pi\in {\mathbb{C}}^{1\times n}$, $\eta\in
{\mathbb{C}}^{n\times 1}$ and $M(\sigma)\in{\mathbb{C}}^{n\times n}$ for
$\sigma\in \Sigma$. The {\em word function} $f_{\cal
A}:\Sigma^{*}\rightarrow {\mathbb{C}}$  associated to  ${\cal A}$ is then defined as follows:
\begin{equation}
f_{\cal
A}(x)=\pi M(x_1)\dots M(x_n)\eta,
\end{equation}
 where  $x=x_1\dots x_n\in
\Sigma^{*}$.  Two BLMs ${\cal A}_1$ and ${\cal A}_2$ are said to be equivalent if $f_{{\cal A}_1}(x)=f_{{\cal A}_2}(x)$ for all $x\in\Sigma^{*}$. For this problem, we recall a result from \cite{LQ08}.

\begin{lemma} \label{lemm:BLM}
There exists a polynomial-time algorithm {\em (}running in
time $O((n_1+n_2)^4)${\em )}  that takes  two BLMs ${\cal A}_1$
and ${\cal A}_2$ as inputs and determines whether ${\cal A}_1$ and ${\cal
A}_2$ are equivalent, where $n_1$ and $n_2$ are the  numbers of states of ${\cal A}_1$
and ${\cal A}_2$, respectively.
\end{lemma}

Using this lemma we will obtain the following result.
\begin{theorem}
 It is decidable whether  two pvDFA are comparable.
\end{theorem}
\begin{proof}
Given two pvDFA ${\cal A}$ and ${\cal B}$, it is sufficient to prove that it is decidable     whether ${\cal A}={\cal B}$, and  whether  ${\cal A}<{\cal B}$.

At first we prove that it is decidable whether  ${\cal A}={\cal B}$. Indeed,
suppose that  a pvDFA ${\cal C}=(S, \Sigma, \delta, s_0, S_{a}, S_r)$ recognizes a promise problem $C=(C_{yes}, C_{no})$. We construct now a  BLM: ${\cal C}'=(S, \pi, \{M(\sigma)\}_{\sigma\in\Sigma},  \eta)$, where $\pi$ is an $|S|$-dimensional row vector with $\pi[s_0]=1$ and $\pi[s]=1$ for $s\neq s_0$, $M(\sigma)$ is an $|S|\times |S|$ matrix with $M(\sigma)[s,t]=1$ if $\delta(s, \sigma) = t$ and $0$ otherwise, and  $\eta$ is an $|S|$-dimensional column vector such that
\begin{align*}
\eta[s]=\left\{
          \begin{array}{ll}
            1, & \hbox{if $s\in S_a$;} \\
            2, & \hbox{if $s\in S_r$;} \\
            0, & \hbox{otherwise.}
          \end{array}
        \right.
\end{align*}
To such a  BLM ${\cal C}'$ we can associate  a function $f_{{\cal C}'}: \Sigma^*\rightarrow \{0,1,2\}$
defined as follows: $f_{{\cal C}'}(x)=1$ iff $x\in C_{yes}$, $f_{{\cal C}'}(x)=2$ iff $x\in C_{no}$, and $f_{{\cal C}'}(x)=0$ iff $x\in \Sigma^*\setminus (C_{yes}\cup C_{no})$.

Therefore, two  pvDFA ${\cal A}$ and ${\cal B}$ are equivalent  iff their associated BLMs ${\cal A}'$ and ${\cal B}'$ are equivalent, i.e., $f_{{\cal A}'}(x)=f_{{\cal B}'}(x)$ for all $x\in\Sigma^*$. The latter problem is decidable by Lemma \ref{lemm:BLM}.

As the next we show that it is decidable whether  ${\cal A}<{\cal B}$. Suppose that  a pvDFA ${\cal C}=(S, \Sigma, \delta, s_0, S_{a}, S_r)$ is such that  $ {\cal P}({\cal C})=(C_{yes}, C_{no})$. Let us now consider  DFA ${\cal C}_y=(S, \Sigma, \delta, s_0, S_{a})$  and ${\cal C}_n=(S, \Sigma, \delta, s_0, S_{r})$.  Clearly ${\cal L}({\cal C}_y)=C_{yes}$ and  ${\cal L}({\cal C}_n)=C_{no}$.  These observations can now be used as follows.
Given two pvDFA ${\cal A}$ and ${\cal B}$, we have ${\cal A}<{\cal B}$ iff ${\cal L}({\cal A}_y)\subseteq {\cal L}({\cal B}_y)$ and ${\cal L}({\cal A}_n)\subseteq {\cal L}({\cal B}_n)$.

 It is clear that ${\cal L}({\cal A})\subseteq {\cal L}({\cal B})$ is equivalent to ${\cal L}({\cal A})\cap {\cal L}({\cal B})={\cal L}({\cal A})$.
The later problem is decidable, since it is easy to construct a DFA ${\cal C}$ recognizing ${\cal L}({\cal A})\cap {\cal L}({\cal B})$ and  the equivalence between  DFA  ${\cal C}$ and ${\cal A}$  is decidable.
Therefore, given two DFA ${\cal A}$ and ${\cal B}$, it is decidable whether ${\cal L}({\cal A})\subseteq {\cal L}({\cal B})$.
\end{proof}

\begin{remark}
Note that for any given pvDFA, there exist algorithms to find an equivalent pvDFA which has the smallest number of states among all  pvDFA equivalent to the given one, since a pvDFA can be  considered  as a  special Moore automaton  whose minimization problem is known to be solvable,  see \cite{BH14} for more details.
\end{remark}

\begin{remark}
   If one of the following cases  ${\cal A}={\cal B}$, ${\cal A}<{\cal B}$ or ${\cal A}>{\cal B}$ holds, then we know that two  pvDFA are comparable.  Otherwise, they are not comparable. Suppose that pvDFA ${\cal A}$ has $n_1$ states and pvDFA ${\cal B}$ has $n_2$ states, it takes polynomial time  ($O((n_1+n_2)^4)$) to determine whether  ${\cal A}={\cal B}$. Given two DFA ${\cal C}$  and DFA ${\cal D}$, it takes also polynomial time to find out ${\cal L}({\cal C})\cap {\cal L}({\cal D})$. According to the above theorem, therefore, it takes polynomial time to determine whether two pvDFA are comparable or not.
\end{remark}

\section{State complexity}

Consideration of state complexity is another way to get a deepen insight in to the power of various types of automata \cite{Yu98}.
 In this section we will deal with  the state complexity of pvDFA for promise problems with respect to recognizability and solvability.

For a regular language $L$, we denote by $s(L)$ the number of states of the  minimal DFA to recognize the language $L$. For a promise problem $A=(A_{yes},A_{no})$ that can be recognized by a pvDFA, we denote by $sr(A)$ the number of states of the minimal pvDFA recognizing $A$. For a promise problem $A=(A_{yes},A_{no})$ that can be solved by a pvDFA, we denote by $ss(A)$ the number of states of the minimal pvDFA solving $A$.

In a DFA  ${\cal A}=(S,\Sigma,\delta,s_0, S_a)$, a
state $s$ is said to be {\em distinguishable} from a state $t$ if there is $w\in\Sigma^*$ such that one of the states $\widehat{\delta}(s,w)$ and  $\widehat{\delta}(t,w)$ is accepting, and the other is not. If every two states in DFA ${\cal A}$ are  distinguishable from each other, then ${\cal A}$ is minimal \cite{Hop}.

\begin{theorem}\label{th9}
If a promise problem $A=(A_{yes},A_{no})$ with $A_{yes}\neq \emptyset$ and $A_{no}\neq \emptyset$  can be recognized by a pvDFA, then
\begin{equation}\label{Eq-com}
  \max\{s(A_{yes}),s(A_{no})\}\leq sr(A)\leq s(A_{yes})s(A_{no})-1.
\end{equation}
\end{theorem}

\begin{proof}
Since  $A$ can be recognized by a pvDFA,   according to Theorem \ref{Th-rec},  $A_{yes}$ and $A_{no}$ are regular languages.

Suppose that  $A$ is recognized by a minimal pvDFA ${\cal A}=(S,\Sigma,\delta,s_0, S_a,S_r)$, we have that the regular language $A_{yes}$ can be recognized by the DFA ${\cal A}_y=(S,\Sigma,\delta,s_0, S_a)$ and the regular language $A_{no}$ can be recognized by the DFA ${\cal A}_n=(S,\Sigma,\delta,s_0, S_r)$. Therefore, $|S|\geq s(A_{yes})$ and   $|S|\geq s(A_{no})$. Hence $sr(A)= |S|\geq \max\{s(A_{yes}),s(A_{no})\}$.

Let us assume that  $A_{yes}$  is recognized by a minimal DFA ${\cal A}^1=(S^1,\Sigma,\delta^1,s_{0}^1, S_{a}^1)$ and $A_{no}$  is recognized by a minimal DFA ${\cal A}^2=(S^2,\Sigma,\delta^2,s_{0}^2, S_{a}^2)$.
 According to  Theorem \ref{Th-rec},
 the promise problem can be recognized by the  pvDFA ${\cal A}=(S,\Sigma,\delta,s_{0}, S_{a},S_r)$ where $S=(S_1\times S_2)\setminus (S_{a}^1\times S_{a}^2)$, $s_{0}=\langle s_{0}^1,  s_{0}^2\rangle$, $\delta(\langle s^1,s^2\rangle, \sigma)=\langle \delta^1(s^1,\sigma),  \delta^2(s^2,\sigma)\rangle$, $S_a=S_{a}^1\times (S^2\setminus S_{a}^2 )$ and $S_r=(S^1\setminus S_{a}^1 ) \times  S_{a}^2$. Therefore,  we have $sr(A)\leq |S|- |S_a^1\times S_a^2|\leq S_1\times S_2-1=s(A_{yes})s(A_{no})-1$.

\end{proof}

A natural problem is whether Inequalities (\ref{Eq-com}) are tight. Next we try to answer them partially. First, we consider the left side.

\begin{theorem}
 The left side of Inequalities (\ref{Eq-com}) is tight.
\end{theorem}

\begin{proof}
We prove that $sr(A)= \max\{s(A_{yes}),s(A_{no})\}$ in some cases. Let us consider the promise problem $A^{N,\,l}=(A_{yes}^{N,\,l},A_{no}^{N,\,l})$ with  $A_{yes}^{N,\,l}=\{a^{iN}\,|\,\ i\geq 0\}$ and $A_{no}^{N,\,l}=\{a^{iN+l}\,|\,\ i\geq 0\}$, where $N$ is a fix prime and $l$ is a fix positive integer such that $0< l<N$. It is easy to see that $s(A_{yes}^{N,\,l})=N$ and $s(A_{no}^{N,\,l})=N$.
Let us consider now an $N$-state pvDFA ${\cal B}=(S,\{a\},\delta,s_0, S_a,S_r)$, where $S=\{s_0,s_1,\ldots,s_{N-1}\}$, $S_a=\{s_0\}$, $S_r=\{s_l\}$ and $\delta(s_i,a)=s_{(i+1)\ {\it mod}\  N}$. It is easy to check that the promise problem $A^{N,\,l}$ can be recognized by the pvDFA ${\cal B}$.

Let us assume now that the promise problem $A^{N,\,l}$ can be recognized by an $M$-state pvDFA ${\cal B}'=(S',\{a\},\delta',s'_0, S'_a,S'_r)$ and $M<N$. It is easy to see that
the DFA ${\cal B}'_1=(S',\{a\},\delta',s'_0, S'_a)$ can solve the promise problem $A^{N,\,l}$. Therefore, the minimal DFA to solve the promise problem  $A^{N,\,l}$ has less  than $N$ states,  contradicting   the fact that
 the minimal DFA to solve  $A^{N,\,l}$ has $N$ states \cite{GQZ14b}.

Therefore,  $sr(A^{N,\,l})=\max\{s(A_{yes}^{N,\,l}),s(A_{no}^{N,\,l})\}=N$.
\end{proof}

For the right side, we only know the following relation.

\begin{theorem}\label{Bound}
 There is a promise problem $A$ satisfying  $sr(A)=\frac{1}{2}s(A_{yes})s(A_{no})$.
\end{theorem}

\begin{proof}

In the interest of readability, we put the proof in Appendix.

\end{proof}

\begin{theorem}
If a promise problem $A=(A_{yes},A_{no})$ can be recognized by a pvDFA, then $ss(A)\leq \min\{s(A_{yes}),\linebreak[0]s(A_{no})\}$.
\end{theorem}
\begin{proof}
%Suppose the promise problem $A$ can be recognized by pvDFA ${\cal A}=(S,\Sigma,\delta,s_0, S_a,S_r)$.

 According to Theorem \ref{Th-rec},   $A_{yes}$ and $A_{no}$ are regular languages. Suppose  $A_{yes}$ can be recognized by a minimal DFA ${\cal A}^1=(S^1,\Sigma^1,\delta^1,s_0^1, S_a^1)$. This implies  that the promise problem $A$ can be solved by the DFA ${\cal A}^1$ and therefore $ss(A)\leq s(A_{yes})$.
  Suppose  $A_{no}$ can be recognized by a minimal DFA ${\cal A}^2=(S^2,\Sigma^2,\delta^2,s_0^2, S_a^2)$. We get that the promise problem $A$ can be solved by the DFA  ${\cal A}^2$ and therefore $ss(A)\leq s(A_{no})$.
  Hence $ss(A)\leq \min\{s(A_{yes}),s(A_{no})\}$.

We prove that  $ss(A)= \min\{s(A_{yes}),s(A_{no})\}$ in same cases.  Let us consider the promise problem $A^{N,\,l}=(A_{yes}^{N,\,l},A_{no}^{N,\,l})$ with  $A_{yes}^{N,\,l}=\{a^{iN}\,|\,\ i\geq 0\}$ and $A_{no}^{N,\,l}=\{a^{iN+l}\,|\,\ i\geq 0\}$, where $N$ is a fix prime and $l$ is a positive integer such that $0< l<N$. It is easy to see that $s(A_{yes}^{N,\,l})=N$ and $s(A_{no}^{N,\,l})=N$. It has been proved in \cite {GQZ14b} that  $ss(A^{N,\,l})=N$. Therefore  $ss(A^{N,\,l})=\min\{s(A_{yes}^{N,\,l}), s(A_{no}^{N,\,l})\}=N$.
\end{proof}

\begin{remark}
 $ss(A)$ can be very small with respect to $s(A_{yes})$ and $s(A_{no})$. For example, let us  consider the promise problem $A^{N,\,l}=(A_{yes}^{N,\,l},A_{no}^{N,\,l})$ with  $A_{yes}^{N,\,l}=\{a^{iN}\,|\,\ i\geq 0\}$ and $A_{no}^{N,\,l}=\{a^{iN+l}\,|\,\ i\geq 0\}$, where $N$ is a fix even integer and $l$ is fix  odd integer such that $0< l<N$. Obviously, we have $s(A_{yes}^{N,\,l})=N$ and $s(A_{no}^{N,\,l})=N$.  However $ss(A^{N,\,l})=2$, since the length of the input $|w|$ is even if $w\in A_{yes}^{N,\,l}$ and the length of the input $|w|$ is odd if $w\in A_{no}^{N,\,l}$.
\end{remark}

\section{One-way quantum finite automata  for promise problems}\label{s-5}

It has been proved that two-way quantum finite automata (2QFA) \cite{Kon97} and also 2QCFAs \cite{Amb02} are more powerful than two-way probabilistic finite automata (2PFA) in recognizing languages.  2QCFA are also more powerful than 2PFA in solving promise problems \cite{RY14}.
In the case of one-wayness, it has been proved that one-way quantum finite automata (1QFA) are not more powerful than one-way  classical finite automata (1FA) \cite{Amb98,Hir10,LiQiu09} in recognizing languages. However, we will prove that 1QFA can be more powerful than their classical counterparts when  recognizing promise problems.

 We  prove now that the exact 1QFA have advantages in recognizing promise problems comparing to their classical counterparts (DFA). Some of the proof techniques can be found in \cite{GQZ14b}.

Let us   consider  a  family of  promise problems $A^{l}=(A_{yes}^{l},A_{no}^{l})$ with $A_{yes}^{l}=\{w\in\{a,b\}^*\mid \#_a(w)= \#_b(w)\}$ and $A_{no}^{l}=\{w\in\{a,b\}^*\mid \#_a(w)+l= \#_b(w)\}$, where $l$ is a fix positive integer such that $ (2\pi i+\frac{\pi}{2}) \leq \sqrt{2}l\leq (2\pi i+\frac{3\pi}{2})$ for some integer $i$.

\begin{theorem} \label{ThMOQ-DFA}
The promise problems  $A^{l}$ can be recognized exactly by a pvMO-1QFA and can not be recognized by any pvDFA.
\end{theorem}
 \begin{proof}
Let
    \begin{equation}
    \theta=\sqrt{2}\pi, p=\cos l\theta, \alpha= \sqrt{\frac{-p}{1-p}}=\sqrt{\frac{-\cos l\theta}{1-\cos l\theta}} \mbox{ and } \beta=\sqrt{\frac{1}{1-p}}=\sqrt{\frac{1}{1-\cos l\theta}}.
    \end{equation}
     We will now construct a pvMO-1QFA ${\cal M}^l=(Q,\{a,b\},\{U_{\sigma}\,|\, \sigma\in\{|\hspace{-1mm}c, a,b,\$\} \},|{0}\rangle,Q_a,Q_r)$ to recognize  $A^{l}$ exactly, where
 \begin{itemize}
   \item $Q=\{|{0}\rangle,|{1}\rangle,|{2}\rangle\}$, $Q_a=\{|{0}\rangle\}$, $Q_r=\{|{1}\rangle, |{2}\rangle\}$.
   \item $U_{\sigma}$ are defined as follows:
    \begin{equation}
    U_{|\hspace{-1mm}c}=\left(
      \begin{array}{ccc}
        \alpha \ & -\beta \ & 0  \\
        \beta \ & \alpha \ & 0   \\
        0 \ &  0  \  &1\\
      \end{array}
    \right),\
    U_{a}=\left(
      \begin{array}{ccc}
        1 & 0 & 0  \\
        0 & \cos \theta & \sin \theta   \\
        0 &  -\sin \theta  &\cos \theta\\
      \end{array}
    \right) ,\
    U_{b}=\left(
      \begin{array}{ccc}
        1 & 0 & 0  \\
        0 & \cos \theta & -\sin \theta   \\
        0 &  \sin \theta  &\cos \theta\\
      \end{array}
    \right), \
    U_{\$}=U_{|\hspace{-1mm}c}^{-1}.
\end{equation}
 \end{itemize}

 See \cite{GQZ14b} for more intuitions  why we choose $U_{|\hspace{-1mm}c}$ and $U_{\$}$ in the way as above.   Since $U_{a}U_{b}=U_{b}U_{a}=I$, for $w=\sigma_1\ldots \sigma_{|w|}\in\{a,b\}^*$, we have
 \begin{equation}
 U_{w}=U_{\sigma_{|w|}} \ldots U_{\sigma_1}=U_{a}^{\#_a(w)}U_{b}^{\#_b(w)}.
 \end{equation}

 Let $\#_a(w)=n$ and  $\#_b(w)=m$.
 If   $w\in A^{l}_{yes}$, then the quantum state before the measurement is
\begin{equation}
    |q\rangle=U_{\$}U_wU_{|\hspace{-1mm}c}|0\rangle=U_{\$}(U_a)^{n}(U_b)^{m}U_{|\hspace{-1mm}c}|0\rangle=U_{\$}(U_a)^{n}(U_b)^{n}U_{|\hspace{-1mm}c}|0\rangle=|0\rangle
\end{equation}
and if the input  $w\in A^{l}_{no}$, then the quantum state before the measurement is
\begin{equation}
    |q\rangle=U_{\$}U_wU_{|\hspace{-1mm}c}|0\rangle=U_{\$}(U_a)^{n}(U_b)^{m}U_{|\hspace{-1mm}c}|0\rangle=U_{\$}(U_a)^{n}(U_b)^{n+l}U_{|\hspace{-1mm}c}|0\rangle=   U_{\$}(U_b)^{l}U_{|\hspace{-1mm}c}|0\rangle=\gamma_1|1\rangle+\gamma_2|2\rangle,
\end{equation}
 where  $\gamma_1$ and $\gamma_2$ are  amplitudes that we do not need to specify more exactly.

Since the amplitude at  $|0\rangle$ in the above quantum state $|q\rangle $ is 0, we  get the exact result after the measurement of $\gamma_1|1\rangle+\gamma_2|2\rangle$ in the standard basis $\{|0\rangle,|1\rangle,|2\rangle\}$. Therefore, we have
\begin{itemize}
  \item if $w\in A^{l}_{yes}$, then $Pr[{\cal M}^l\  \text{accepts}\  w]=1$;
  \item if $w\in A^{l}_{no}$, then $Pr[{\cal M}^l\ \text{rejects}\  w]= 1$.
\end{itemize}

We now give the proof for the other direction. Namely, we show that
  $Pr[{\cal M}^l\  \text{accepts}\  w]=1$ implies that  $w\in A^{l}_{yes}$.

Assume that $w\not\in A^{l}_{yes}$, that is $\#_a(w)\neq \#_b(w)$.  The quantum state before the measurement is
\begin{align}
    |q\rangle&=U_{\$}U_wU_{|\hspace{-1mm}c}|0\rangle=U_{\$}(U_a)^{n}(U_b)^{m}U_{|\hspace{-1mm}c}|0\rangle=  U_{\$}(U_b)^{m-n}U_{|\hspace{-1mm}c}|0\rangle\\
    &=\left(
      \begin{array}{ccc}
        \alpha \ & \beta \ & 0  \\
        -\beta \ & \alpha \ & 0   \\
        0 \ &  0  \  &1\\
      \end{array}
    \right)
    \left(
      \begin{array}{ccc}
        1 & 0 & 0  \\
        0 & \cos (m-n)\theta & -\sin (m-n)\theta   \\
        0 &  \sin (m-n)\theta  &\cos (m-n)\theta\\
      \end{array}
    \right)
    \left(
      \begin{array}{ccc}
        \alpha \ & -\beta \ & 0  \\
        \beta \ & \alpha \ & 0   \\
        0 \ &  0  \  &1\\
      \end{array}
    \right)
    \left(
      \begin{array}{c}
        1   \\
       0    \\
        0 \\
      \end{array}
    \right)\\
    &=\left(
      \begin{array}{c}
       \alpha^2+ \beta^2\cos (m-n) \theta \\
       -\alpha\beta+ \alpha\beta \cos (m-n) \theta  \\
        \beta\sin (m-n)\theta  \\
      \end{array}
    \right).
 \end{align}
Since $\theta=\sqrt{2}\pi$, there are no integers $m\neq n$ such that $\cos (m-n) \theta=1$.  Therefore $ \alpha^2+ \beta^2\cos (m-n) \theta \neq 1$ and $Pr[{\cal M}^l\  \text{accepts}\  w]\neq 1$.

 We now prove the following:
If  $Pr[{\cal M}^l\  \text{rejects}\  w]=1$, then the input $w\in A^{l}_{no}$.

Assume that $w\not\in A^{l}_{no}$, that is $\#_a(w)\neq \#_b(w)+l$.  The quantum state before the measurement is
\begin{align}
    |q\rangle&=U_{\$}U_wU_{|\hspace{-1mm}c}|0\rangle=U_{\$}(U_a)^{n}(U_b)^{m}U_{|\hspace{-1mm}c}|0\rangle=  U_{\$}(U_b)^{m-n}U_{|\hspace{-1mm}c}|0\rangle=\left(
      \begin{array}{c}
       \alpha^2+ \beta^2\cos (m-n) \theta \\
       -\alpha\beta+ \alpha\beta \cos (m-n) \theta  \\
        \beta\sin (m-n)\theta  \\
      \end{array}
    \right).
 \end{align}
Let $m-n=l'$. Since $\theta=\sqrt{2}\pi$ and $m\neq n+l$,  we have

\begin{equation}
\alpha^2+ \beta^2\cos (m-n) \theta =\alpha^2+ \beta^2\cos l' \theta=\frac{-\cos l\theta}{1-\cos l\theta}+\frac{1}{1-\cos l\theta}\cos l' \theta=\frac{\cos l' \theta -\cos l\theta}{1-\cos l\theta}\neq 0.
\end{equation}
  Therefore,  $Pr[{\cal M}^l\  \text{accepts}\  w]\neq 1$.

Hence, we have proved that the promise problem  $A^{l}$ can be recognized exactly by the pvMO-1QFA  ${\cal M}^l$.  Obviously, $A^{l}_{yes}$ and $A^{l}_{no}$ are not regular languages. According to Theorem \ref{Th-rec}, the promise problem  $A^{l}$ cannot be recognized by any pvDFA.
\end{proof}

\begin{remark}
  From Theorem \ref{ThMOQ-DFA} it implies that there are three subsets (non-regular languages) that can be distinguished precisely by a pvMO-1QFA, but any pvDFA cannot do it, and this result further shows a stronger aspect of 1QFA than DFA.
\end{remark}

We will now consider
solvability mode.
Geffert and Yakary{\i}lmaz \cite{GY14} proved that the promise problem \mbox{ExpEQ}($c$)\footnote{
$\mbox{ExpEQ}(c)= \left\{\begin{array}{ll}
\mbox{ExpEQ}_{yes}(c)=\{(a^mb^n)^{3 (2c^2)^{m+n}\cdot\lceil\ln c\rceil}\mid m,n \in \mathbf{N}_{+}, m=n\} \\
\mbox{ExpEQ}_{no}(c)\ =\{(a^mb^n)^{3 (2c^2)^{m+n}\cdot\lceil\ln c\rceil}\mid m,n \in \mathbf{N}_{+}, m\neq n\}
                  \end{array}
                   \right.,
                  $
 where  $c\geq 3$ is an integer.
}
can be solved by a {\em one-way probability finite automaton} (PFA) ${\cal A}(c)$, but there is  no DFA solving \mbox{ExpEQ}($c$).
Rashid and Yakary{\i}lmaz \cite{RY14} proved that a promise problem can be solved by  a Las Vegas realtime rtQCFA or by  an exact rational restarting rtQCFA in linear expected time, where there is no bounded-error PFA that solves the promise problem.
In order to prove that 1QCFA have advantages  in solving promise problems comparing to their classical counterparts (PFA), we define a new promise problem
\begin{equation}
\mbox{PloyEQ}=\left\{\begin{array}{ll}
                      \mbox{PloyEQ}_{yes}=\{(a^{n}b^{m} \#)^t \mid  n=m \ \mbox{and}\ t\geq T\}, \\
                    \mbox{PloyEQ}_{no}\ =\{(a^{n}b^{m} \#)^t   \mid n\neq m \ \mbox{and}\ t\geq T\},
                  \end{array}
 \right.
 \end{equation}
where $T$ is a polynomial of $l=\max\{n,m\}$ which will be specified later.
 %\begin{equation}
%\mbox{PloyEQ}=\left\{\begin{array}{ll}
 %                     \mbox{PloyEQ}_{yes}=\{a^{n_1}b^{n_1} \#a^{n_2}b^{n_2}\#\ldots  a^{n_t}b^{n_t}  \# \mid t\geq T\} \\
 %                   \mbox{PloyEQ}_{no}\ =\{a^{n_1}b^{m_1} \#a^{n_2}b^{m_2}\# \ldots  a^{n_t}b^{m_t}\#   \mid n_i\neq m_i \ \mbox{and}\ t\geq T\}
 %                 \end{array}
 % \right.,
 %\end{equation}

 \begin{theorem} \label{ThQCFA-PFA}
For any $\varepsilon\leq \frac{1}{3}$,  the  promise problem $\mbox{PloyEQ}$ can be  solved by  a 1QCFA  with the error probability  $\varepsilon$, but there is no PFA solving  $\mbox{PloyEQ}$ with the error probability  $\varepsilon$.
\end{theorem}

 \begin{proof}
 Let $\theta=\sqrt{2}\pi$.
We design a   1QCFA ${\cal M}=(Q,S,\Sigma,\Theta,\Delta,\delta,|q_{0}\rangle,s_{0},S_{a},S_{r})$  to solve the promise problem $\mbox{PloyEQ}$, where $Q=\{|{0}\rangle,|{1}\rangle\}$.   The automaton ${\cal M}$ proceeds as shown in Figure \ref{f3}, where
    \begin{equation}
    U_{|\hspace{-1mm}c}= U_{\$}=I,\
    U_{a}=\left(
      \begin{array}{cc}
       \cos \theta & \sin \theta   \\
        -\sin \theta  &\cos \theta\\
      \end{array}
    \right) ,\
    U_{b}=\left(
      \begin{array}{cc}
      \cos \theta & -\sin \theta   \\
      \sin \theta  &\cos \theta\\
      \end{array}
    \right).
     \end{equation}

 \begin{figure}[htbp]
 %  %Requires \usepackage{graphicx}
\begin{tabular}{|l|}
    \hline

\begin{minipage}[t]{0.8\textwidth}
\begin{enumerate}
\item[1.] Read the left end-marker $|\hspace{-1.5mm}c$,  perform $U_{|\hspace{-1mm}c}=I$ on the initial quantum state $|0\rangle$, do not change its classical state, and move the tape head one cell to the right.

\item[2.] Until the currently  scanned symbol $\sigma$ is  the right end-marker $\$$, do the following:
 \begin{enumerate}
 \item[2.1] If $ \sigma\neq \#$, apply $\Theta(s_0,\sigma)=U_{\sigma}$ to the current quantum state, do not change its classical state, and move the tape head one cell to the right.
 \item[2.2] Otherwise, measure the current quantum state with $M=\{|0\rangle\langle 0|,|1\rangle\langle 1| \}$.   If the outcome is $|1\rangle$, reject the input and halt. Otherwise, move the tape head one cell to the right.
\end{enumerate}
\item[3.] Accept the input and halt.

\end{enumerate}

\end{minipage}\\

\hline
\end{tabular}
 \centering\caption{  The 1QCFA  solving the promise problem $\mbox{PloyEQ}$.  }\label{f3}
\end{figure}

 Let us choose $T=\lceil 2l^2\log_{e}\frac{1}{\varepsilon}\rceil$.
 If the input $w\in \mbox{PloyEQ}_{yes}$, then the quantum state before the measurement in the Step 2.2 is always $|0\rangle$. Therefore, the input will be accepted with certainty.

  If the input $w\in \mbox{PloyEQ}_{no}$, the quantum state before the $i$-th measurement in the Step 2.2 is
  \begin{align}
  |q\rangle&=U_{a}^{n}
    U_{b}^{m}=\left(
      \begin{array}{cc}
       \cos \theta & \sin \theta   \\
        -\sin \theta  &\cos \theta\\
      \end{array}
    \right)^{n}
\left(
      \begin{array}{cc}
      \cos \theta & -\sin \theta   \\
      \sin \theta  &\cos \theta\\
      \end{array}
    \right)^{m}\\
   & =\left(
      \begin{array}{cc}
      \cos (m-n)\theta & -\sin (m-n)\theta   \\
      \sin (m-n)\theta  &\cos (m-n)\theta\\
      \end{array}
    \right).
   \end{align}

   According to \cite {Amb02,Zhg12}, the rejecting probability after the $i$-th measurement is
   \begin{equation}
P_{ir} >\frac{1}{2(m-n)^2+1}>\frac{1}{2l^2}.
\end{equation}
and the overall probability that ${\cal M}$ rejects the input $w$  is
\begin{align}
Pr[{\cal M}\ \text{rejects}\  w]&=    \sum_{i\geq
1}^{t}\left(P_{ir}\prod^{i-1}_{i=1}  (1-P_{r(i-1)})\right)
> \sum_{i\geq
1}^{t}\left(\frac{1}{2l^2}\prod^{i-1}_{i=1}  (1-\frac{1}{2l^2})\right)\\
&=\sum_{i\geq
1}^{t}\frac{1}{2l^2} (1-\frac{1}{2l^2})^{i-1}=\frac{1}{2l^2} \frac{1-(1-\frac{1}{2l^2})^t}{\frac{1}{2l^2}}=1-(1-\frac{1}{2l^2})^t.
\end{align}

Since $1-x\leq e^{-x}$, we have
\begin{align}
Pr[{\cal M}\ \text{rejects}\  w]>1-(1-\frac{1}{2l^2})^t>1-e^{-\frac{1}{2l^2}t}\geq 1-e^{-\frac{1}{2l^2} 2l^2\log_{e}\frac{1}{\varepsilon}}=1-e^{-\log_{e}\frac{1}{\varepsilon}}=1-\varepsilon.
\end{align}

Therefore, the promise problem $\mbox{PloyEQ}$  can be solved by a 1QCFA ${\cal M}$  with the error probability  $\varepsilon$.

Assume now that there is a PFA ${\cal A}$ solving  $\mbox{PloyEQ}$ with the error probability $\varepsilon$. Let us consider a 2PFA ${\cal M}$ running as follows:
\begin{enumerate}
  \item ${\cal M}$ reads the input $w$ from the left to the right -- symbol by symbol;
  \item After reading each  $\sigma\in\{a,b,\#\}$,  ${\cal M}$ simulates the transformation of the PFA ${\cal A}$ reading $\sigma$;
  \item When ${\cal M}$ reaches the right-end marker,  ${\cal M}$ moves its   tape head   to the left most  symbol of the input $w$ and  reads the input $w$ again.
\end{enumerate}

If  ${\cal M}$ reads the input $w$  $T$  times,  then we have,  according to the above assumption,
\begin{equation}
Pr[{\cal M}\ \mbox{accepts}\  a^nb^n\#]=Pr[{\cal A}\ \mbox{accepts}\  a^nb^n\#]\geq 1-\epsilon
\end{equation}
and
\begin{equation}
Pr[{\cal M}\ \mbox{accepts}\  a^nb^m\#]=Pr[{\cal A}\ \mbox{accepts}\  a^nb^m\#]\leq 1-Pr[{\cal A}\ \mbox{rejects}\  a^nb^m\#] \leq \epsilon
\end{equation}
where $n\neq m$.

Therefore, for any integers $n$ and $d>0$, it holds
\begin{equation}\label{Eq-c}
 \left|Pr[{\cal M}\ \mbox{accepts}\  a^nb^n\#]-Pr[{\cal M}\ \mbox{accepts}\  a^nb^{n+d}\#] \right|\geq 1-2\epsilon\geq \epsilon.
\end{equation}

Since $T$ is a polynomial of the length of the input $w$,  the following  lemma holds (as  in \cite{GW86,DwS92}):
\begin{lemma}\label{Lm-GW86}
Let $\varepsilon\leq\frac{1}{3}$. Suppose that ${\cal M}$ is a  two-way probabilistic finite automaton (2PFA) with $\mbox{exp}(o(|w|))$ expected running time, where $|w|$ is the length of the input. Then there exists, for all sufficiently large $n$,  an integer $d$ such that
\begin{equation}\label{Eq-gw86}
    \left|Pr[{\cal M}\ \mbox{accepts}\  a^nb^n\#]-Pr[{\cal M}\ \mbox{accepts}\  a^nb^{n+d}\#] \right| <\epsilon.
\end{equation}
\end{lemma}

Obviously, Equality   (\ref{Eq-c})   contradicts Equality(\ref{Eq-gw86}).
Therefore,  there is no PFA  solving $\mbox{PloyEQ}$ with  the error probability  $\varepsilon$.
\end{proof}

We now study state complexity.  We consider the following promise problem
\begin{equation}
A(p)=\left\{\begin{array}{ll}
                     A_{yes}(p)=\{a^{ip+l_1}\,|\, 0\leq l_1<p, \cos^2 l_1 \theta \geq 2/3, \ i\geq 0\}, \\
                    A_{no}(p)\ =\{a^{ip+l_2}\,|\, 0\leq l_2<p, \cos^2 l_2 \theta \leq 1/3, \ i\geq 0\},
                  \end{array}
 \right.
 \end{equation}
where $\theta=\pi/p$.

\begin{theorem} \label{ThMOQ-MOQ}
For integer $p\geq 6$,
the promise problems  $A(p)$ can be solved with error probability $\epsilon\leq 1/3$ by an  MO-1QFA with two quantum basis states, but can not be solved exactly by any MO-1QFA with two quantum basis states.
\end{theorem}
\begin{proof}
We will now construct an MO-1QFA ${\cal M}(p)=(Q,\{a\},\{U_{\sigma}\,|\, \sigma\in\{|\hspace{-1mm}c, a,\$\} \},|{0}\rangle,Q_a)$ to solve  $A(p)$, where
 \begin{itemize}
   \item $Q=\{|{0}\rangle,|{1}\rangle\}$, $Q_a=\{|{0}\rangle\}$.
   \item $U_{\sigma}$ are defined as follows:
    \begin{equation}
    U_{|\hspace{-1mm}c}=U_{\$}=I,\
    U_{a}=\left(
      \begin{array}{ccc}
        \cos \theta & -\sin \theta   \\
        \sin \theta  &\cos \theta\\
      \end{array}
    \right).
\end{equation}
 \end{itemize}

If  input $w\in A_{yes}(p)$, then the quantum state before the measurement is
\begin{equation}
    |q\rangle=U_{\$}U_wU_{|\hspace{-1mm}c}|0\rangle=U_{\$}(U_a)^{ip+l_1}U_{|\hspace{-1mm}c}|0\rangle=(U_a)^{l_1}|0\rangle=\cos l_1\theta|0\rangle+\sin l_1\theta|1\rangle.
\end{equation}
The automaton ${\cal M}$ has the accepting probability
\begin{equation}
Pr[{\cal M}\ \mbox{accepts}\  w] =\cos^2 l_1\theta\geq 2/3.
\end{equation}
If  input $w\in A_{no}(p)$, then the quantum state before the measurement is
\begin{equation}
    |q\rangle=U_{\$}U_wU_{|\hspace{-1mm}c}|0\rangle=U_{\$}(U_a)^{ip+l_2}U_{|\hspace{-1mm}c}|0\rangle=(U_a)^{l_2}|0\rangle=\cos l_2\theta|0\rangle+\sin l_2\theta|1\rangle.
\end{equation}
The automaton ${\cal M}$ has the rejecting probability
\begin{equation}
 Pr[{\cal M}\ \mbox{rejects}\  w]=1- Pr[{\cal M}\ \mbox{accepts}\  w]=1-\cos^2 l_2\theta\geq 1-1/3=2/3.
\end{equation}
Therefore, $A(p)$ can be solved by the automaton ${\cal M}$ with error probability $1/3$.

Suppose that the promise problems  $A(p)$ can be solved exactly by an  MO-1QFA ${\cal M}'$ with two basis states. Without loss of generality, we assume that ${\cal M}'=(Q,\{a\},\{U_{\sigma}\,|\, \sigma\in\{|\hspace{-1mm}c, a,\$\} \},|{0}\rangle,Q_a)$, where $Q=\{|0\rangle,|1\rangle\}$ and $Q_a=\{|0\rangle\}$.

Since $p\geq 6$, we have $a^p\in A_{yes}(p)$ and $a^{p+1}\in A_{yes}(p)$. Since the probability that ${\cal M}'$ accepts $a^p$ is 1, we have
\begin{equation}
   U_{\$}(U_a)^pU_{|\hspace{-1mm}c}|0\rangle=\alpha|0\rangle,
\end{equation}
where $\alpha\in \mathbb{C}$ and $|\alpha|=1$. Therefore,  $(U_a)^pU_{|\hspace{-1mm}c}|0\rangle=\alpha U_{\$}^{\dag}|0\rangle$, where $U^{\dag}$ is conjugate and transpose of $U$.   Since the probability that ${\cal M}'$ accepts $a^{p+1}$ is also 1, we have also
\begin{equation}
   U_{\$}(U_a)^{p+1}U_{|\hspace{-1mm}c}|0\rangle=\alpha'|0\rangle,
\end{equation}
where $\alpha'\in \mathbb{C}$ and $|\alpha'|=1$. Therefore, we have
\begin{align}
   & U_{\$}U_a(U_a)^{p}U_{|\hspace{-1mm}c}|0\rangle=U_{\$}U_a \cdot \alpha U_{\$}^{\dag}|0\rangle  =\alpha'|0\rangle,\\
   &\Rightarrow U_{\$}U_a   U_{\$}^{\dag}|0\rangle  =\frac{\alpha'}{\alpha}|0\rangle.
\end{align}
It is easy to find out that
\begin{equation}
   U_{\$}U_a   U_{\$}^{\dag}=\left(
      \begin{array}{ccc}
        \beta_1 & 0   \\
        0  &\beta_2\\
      \end{array}
    \right)=\Lambda,
\end{equation}
    where $\beta_1=\frac{\alpha'}{\alpha}$ and $\beta_2\in \mathbb{C}$ with $|\beta_2|=1$. It is easy to see that $|\beta_1|=1$. Therefore, we have
   $U_a=U_{\$}^{\dag}\Lambda U_{\$}$.

   Now for any integer $k\geq 0$, we have
   \begin{equation}
   U_{\$}(U_a)^{p+k}U_{|\hspace{-1mm}c}|0\rangle=U_{\$}(U_a)^k(U_a)^pU_{|\hspace{-1mm}c}|0\rangle=U_{\$}(U_{\$}^{\dag}\Lambda U_{\$})^k\cdot \alpha U_{\$}^{\dag}|0\rangle =\alpha \Lambda^k |0\rangle=\alpha \beta_1^k|0\rangle.
\end{equation}
Obviously, $|\alpha \beta_1^k|=1$.  Therefore, for any $k\geq 0$, the automaton accepts the input $a^{p+k}$ with probability 1. If $k=\lceil p/2\rceil$, it is easy to check that $\cos ^2 k\theta\leq 1/3$ and $a^{p+k}\in A_{no}(p)$. Thus, we get a contradiction.
Therefore, the promise problems  $A(p)$ can not be solved exactly by any  MO-1QFA with two quantum basis states.
\end{proof}

\begin{remark}
In the previous theorem, the error probability $\varepsilon=1/3$. For  $p\geq \frac{\pi}{\arccos \sqrt{1-\varepsilon}}$,  using the same method as the previous theorem, we can prove that the following promise problem
\begin{equation}
A(p,\varepsilon)=\left\{\begin{array}{ll}
                     A_{yes}(p,\varepsilon)=\{a^{ip+l_1}\,|\, 0\leq l_1<p, \cos^2 l_1 \theta \geq 1-\varepsilon, \ i\geq 0\}, \\
                    A_{no}(p,\varepsilon)\ =\{a^{ip+l_2}\,|\, 0\leq l_2<p, \cos^2 l_2 \theta \leq \varepsilon, \ i\geq 0\},
                  \end{array}
 \right.
\end{equation}
where $\theta=p/\pi$,
 can be solved with error probability $\varepsilon$ by an MO-1QFA with two quantum  basis states, but can not be solved exactly by any MO-1QFA with two quantum basis states.
\end{remark}

We consider now the minimal PFA to solve the promise problem $A(p)$ with $p$ is prime.

\begin{theorem} \label{ThMOQ-PFA}
For any prime $p>6$,
the minimal PFA solving the promise problem  $A(p)$ with error probability (smaller than $1/2$) has $p$ states.
\end{theorem}
\begin{proof}
We consider now a $p$-state DFA ${\cal A}=(S,
\{a\},\delta,s_0,S_a)$, with the set of states   $S=\{s_0,s_1,\ldots,s_{p-1}\}$,  the set of accepting states  $S_a=\{s_{l_1}\,|\,  0\leq l_1<p, \cos^2 l_1 \theta \geq 2/3 \}$, and the transition function $\delta(s_i,a)=s_{(i+1)\ {\it mod}\  p}$.
Obviously, the promise problem $A(p)$ can be solved by the automaton ${\cal A}$. A DFA is also a PFA.  Therefore, there is a PFA with $p$ states solving the promise problem $A(p)$. The minimal PFA that solving the promise problem $A(p)$ has not more than $p$ states.

Since $p>6$, there must be fix integers $r_1,r_2$ such that $\cos^2 r_1 \theta \geq 2/3 $ and $\cos^2 r_2 \theta \leq 1/3$.  We consider the following promise problem  \cite{GQZ14b}.  Namely, $A^{N,r_1,r_2}=(A_{yes}^{N,r_1},A_{no}^{N,r_2})$ with  $A_{yes}^{N,r_1}=\{a^{n}\,|\,\ n \equiv r_1\ {\it mod}\  N\}$ and $A_{no}^{N,r_2}=\{a^{n}\,|\,\ n \equiv r_2\ {\it mod}\  N\}$, where $N$, $r_1$ and $r_2$ are  fixed positive integers such that $r_1\not\equiv r_2\ {\it mod}\  N$. Let $N=p$ and $l=(r_2-r_1)\mod p$. According to Subsection 3.3, we have $A^{p,r_1,r_2}\leq A(p)$.
Any PFA that solving the promise problem $A(p)$ can also solve the promise problem  $A^{p,r_1,r_2}$. According to \cite{BMP14} (see Theorem 4), the minimal PFA solving the promise $A^{p,r_1,r_2}$ with error probability  has $d$ states, where $d$ is the smallest positive integer such that
$d \mid p$ and $d\nmid l$. Since $p$ is prime, we have $d=p$. Therefore, the minimal PFA that solving the promise problem $A(p)$ has at least $p$ states.
Thus, the theorem has been proved.
\end{proof}

\section{Conclusions and problems}

In order to make clear the difference between recognizability and solvability of quantum and classical finite automata, we have introduced several  promise versions  finite automata and  discussed their properties.
We have explored some basic properties of promise problems recognized and solved by pvDFA, and we have showed the state complexity  for several promise problems concerning recognizability and solvability. In particular, we have proved that one-way quantum finite automata can be more powerful than their classical counterparts when recognizing and solving some promise problems. More specifically,  we have proved:

\begin{itemize}
\item There is a promise problem that can be recognized exactly by {\em  measure-once one-way quantum finite automata} (MO-1QFA), but no {\em deterministic finite automata} (DFA) can recognize it.  Indeed, this result implies that there are three subsets (non-regular languages) that can be distinguished precisely by a pvMO-1QFA, but any pvDFA cannot do it.

      \item  There is a promise problem that can be solved with error probability $\epsilon\leq 1/3$ by {\em one-way finite automaton with quantum and classical states} (1QCFA), but no {\em one-way probability finite automaton} (PFA) can solve it with error probability  $\epsilon\leq 1/3$.

          \item Especially, there are promise problems $A(p)$ with size $p$ that can be solved {\em with any error probability} by MO-1QFA with only two quantum basis states, but they can not be solved {\em exactly} by any MO-1QFA with two quantum basis states; in contrast, the minimal {\em one-way probability finite automaton} (PFA) solving $A(p)$ {\it with any error probability} (usually smaller than $1/2$) has $p$ states.

\end{itemize}

However, there are still some problems to be considered for future research, and we list them in the following.

\begin{enumerate}

\item First we concern a problem related to {\em recognizability}: Suppose that a promise problem $A$ can be recognized  by a quantum (or probabilistic) finite automaton with error probability $\epsilon< 1/2$. Then, for any $\epsilon'< \epsilon$, whether is there a quantum (or probabilistic) finite automaton recognizing  $A$ with error probability $\epsilon'$? For solvability, this problem can be verified positively by using the idea of the languages accepted by PFA with {\em bounded error} (e.g., \cite{Paz71}).

    \item Second is a hierarchic problem for the classes solved by quantum finite automata mentioned in Section 1. Namely, let ${\cal C}(P)_{n}$  denote the class of promise problems solved exactly by an MO-1QFA with $n$ quantum basis states. Then, whether does ${\cal C}(P)_{m}\subset {\cal C}(P)_n$ hold for $m\leq n$?

  \item
  For any given regular language $L$, there is,  according to the Myhill-Nerode theorem, a  method to find out a minimal DFA ${\cal A}$ to recognize $L$ \cite{Yu98} .
  For some specific promise problems, it is possible to find out  minimal DFA (pvDFA) to solve the promise problems \cite{AmYa11,BMP14,GY14,GQZ14b}.
  However it is not clear yet whether there is a general way to find out a minimal pvDFA to solve a given promise problem that can be solved by a pvDFA?

   \item  We have proved that for any $\varepsilon\leq \frac{1}{3}$, the  promise problem $\mbox{PloyEQ}$ can be  solved by  a 1QCFA  with the error probability  $\varepsilon$, but there is no PFA solving  $\mbox{PloyEQ}$ with the error probability  $\varepsilon$ (Theorem \ref{ThQCFA-PFA}).  However, whether is there no PFA solving  $\mbox{PloyEQ}$ with the error probability  $1/3<\varepsilon<1/2$?
       Another challenge is to find out some simpler promise problems to demonstrate the advantage of 1QFA in solving promise problems, since the promise problem $\mbox{PloyEQ}$ is quite complex?

   \item   We have proved that the left side in Inequality (\ref{Eq-com}) is tight. Nevertheless, can we prove that the right side is tight?

\end{enumerate}

%\section{Conclusion and discussion}

 %\begin{question}
%Equivalence of two pvDFA and minimization of a given pvDFA.
%\end{question}

%\begin{question}
%Find out a algorithm to determine the relationship of two given pvDFA  ${\cal A}_1$  and ${\cal A}_2$.
%\end{question}

%If  promise problem $A$ can be solved by pvDFA ${\cal A}_1$ and ${\cal A}_1<{\cal A}_2$, then  promise problem $A$ can be solved by pvDFA ${\cal A}_2$.

%\begin{question}
%Find out a minimal pvDFA to solve a given promise problem $A$ if $A$ can be solved by pvDFA.
%\end{question}

%More questions:

%\begin{enumerate}
%  \item More properties of pvDFA?
 % \item Is pvDFA useful?
 % \item Other corresponding models? PDA?  PFA? QFA?
%\end{enumerate}

\section*{Acknowledgements}
This work  was  partly supported by the National
Natural Science Foundation of China (Nos. 61272058,  61472452, 61572532).

\section*{Appendix. The proof of Theorem \ref{Bound}}

\begin{proof}

Let $A_{yes}=\{(a^p)^*\}$ and $A_{no}=\{(a^q)^*a\}$, where $p,q>2$ are integers such that  $\gcd(p,q)=2$.  We first prove that $A_{yes}\cap A_{no}=\emptyset$.  Since $\gcd(p,q)=2$, there exist integers $k_1$ and $k_2$ such that $p=2k_1$ and $q=2k_2$.  Assume that $A_{yes}\cap A_{no}\neq \emptyset$. There must exist integers $i$ and $j$ such that $(a^p)^i=(a^q)^ja$, i.e. $ip=jq+1$. We have $1=ip-jq=i2k_1-j2k_2=2(ik_1-jk_2)$, which is a contradiction. Therefore, $A_{yes}\cap A_{no}=\emptyset$.

Let us consider now the promise problem $A=(A_{yes},A_{no})$.
Since $A_{yes}$  and $A_{no}$ are  regular languages,    the promise problem $A$ can be recognized by a pvDFA. Let us consider the following pvDFA ${\cal A}=(S,\{a\},\delta,s_0, S_a,S_r)$, where
\begin{itemize}
 \item $S=\{\langle s^1_{k\ {\it mod}\ p}, s^2_{k\ {\it mod}\ q}\rangle\mid k\geq 0\}$;
\item $s_{0}=\langle s_{0}^1,  s_{0}^2\rangle$;
\item $\delta(\langle s^1_i,s^2_j\rangle, a)=\langle s^1_{i+1\ {\it mod}\ p}, s^2_{j+1\ {\it mod}\ q}\rangle$;
\item $S_a= \{\langle s^1_{k\ {\it mod}\ p}, s^2_{k\ {\it mod}\ q}\rangle\mid k\equiv 0\ {\it mod}\ p\}$  and $S_r=\{\langle s^1_{k\ {\it mod}\ p}, s^2_{k\ {\it mod}\ q}\rangle\mid k\equiv 1\ {\it mod}\ q\}$.
\end{itemize}

At first, we prove that $|S|=\frac{1}{2}pq$. Let us assume  that  there exist $0\leq k_1<k_2< \frac{1}{2}pq-1$ such that $\langle s^1_{k_1\ {\it mod}\ p}, s^2_{k_1\ {\it mod}\ q}\rangle=\langle s^1_{k_2\ {\it mod}\ p}, s^2_{k_2\ {\it mod}\ q}\rangle $. This implies $k_1\equiv k_2\ {\it mod}\ p$ and $k_1\equiv k_2\ {\it mod}\ q$. Therefore, $p| (k_2-k_1)$ and $q| (k_2-k_1)$. Since $\gcd(p,q)=2$, we have $\frac{1}{2}pq| (k_2-k_1)$, which is a contradiction. Hence, $|S|\geq \frac{1}{2}pq$.  For any $h\geq \frac{1}{2}pq$, let $h=i\times \frac{1}{2}pq+k$ where $0\leq k< \frac{1}{2}pq$.  Since $p| \frac{1}{2}pq$ and $q| \frac{1}{2}pq$, we have  $\langle s^1_{h\ {\it mod}\ p}, s^2_{h\ {\it mod}\ q}\rangle=\langle s^1_{k\ {\it mod}\ p}, s^2_{k\ {\it mod}\ q}\rangle$. Therefore $|S|=\frac{1}{2}pq$.

Secondly, we prove that  $S_a\cap S_r=\emptyset$. Since $\gcd(p,q)=2$, we have $2| p$ and $2| q$.   Assume that $S_a\cap S_r\neq \emptyset$. In such a case, there must exist integers $i$ and $j$ such that $k=ip$ and $k=jq+1$. Therefore, we have $ip=jq+1$ and $2| (ip-jq)=1$, which is a contradiction.

Moreover, it is easy to see that the promise problem $A$ can be recognized by the pvDFA  ${\cal A}$.  Therefore, $sr(A)\leq |S|= \frac{1}{2}pq$.

Finally, we prove that the
  pvDFA ${\cal A}=(S,\{a\},\delta,s_0, S_a,S_r)$ is minimal. Let us consider DFA ${\cal A}'=(S,\{a\},\delta,s_0, S_a\cup S_r)$.  Obviously, the DFA ${\cal A}'$ recognizes the language  $A_{yes}\cup A_{no}$. We prove now that the DFA ${\cal A}'$ is minimal. Let $F=S_a\cup S_r$ and $n=\frac{1}{2}pq$.   For any $0\leq i<j<n$, we prove that the states $s_i=\langle s^1_{i\ {\it mod}\ p}, s^2_{i\ {\it mod}\ q}\rangle$ and  $s_j=\langle s^1_{j\ {\it mod}\ p}, s^2_{j\ {\it mod}\ q}\rangle$ are distinguishable.  Since $s_i\neq s_j$, at most one of the following two conditions    (1) $j-i\equiv 0 \mod p$ and (2) $j-i\equiv 0 \mod q$ holds. We have therefore the following three cases to consider:

  \begin{enumerate}
    \item The condition (1) holds and (2) does not hold. In such a case we have $\widehat{\delta}(s_i,a^{n-i+1})=\langle s^1_{n+1\ {\it mod}\ p}, \linebreak[0]s^2_{n+1\ {\it mod}\ q}\rangle=\langle s^1_{1}, s^2_{1}\rangle\in F$
    and $\widehat{\delta}(s_j,a^{n-i+1})=\langle s^1_{j+n-i+1\ {\it mod}\ p}, s^2_{j+n-i+1\ {\it mod}\ q}\rangle=\langle s^1_{1}, s^2_{j-i+1\ {\it mod}\ q}\rangle$. Since  $j-i \not\equiv 0\ {\it mod}\ q$, we have $j-i+1 \not\equiv 1\ {\it mod}\ q$. Therefore, $\widehat{\delta}(s_j,a^{n-i+1})\not\in F$.  Hence $s_i$ and $s_j$ are distinguishable.

    \item  The condition (2) holds and (1) does not hold. The proof is similar to the one  in the case 1.

    \item Neither the condition (1) nor (2)  holds. In such a case we have  $\widehat{\delta}(s_i,a^{n-i})=\langle s^1_{n\ {\it mod}\ p}, s^2_{n\ {\it mod}\ q}\rangle=\langle s^1_{0}, s^2_0\rangle\in F$ and  $\widehat{\delta}(s_j,a^{n-i})=\langle s^1_{n+j-i\ {\it mod}\ p}, s^2_{n+j-i\ {\it mod}\ q}\rangle=\langle s^1_{j-i\ {\it mod}\ p}, s^2_{j-i\ {\it mod}\ q}\rangle$. If $\widehat{\delta}(s_j,a^{n-i})\not\in F$, then $s_i$ and $s_j$ are distinguishable. Otherwise, we have $j-i\equiv 1\ {\it mod}\ q$ since $j-i\not\equiv 0\ {\it mod}\ p$. There are now two subcases to consider.
         \begin{enumerate}
        \item    $j-i\equiv 1\ {\it mod}\ p$. We have $p|(j-i-1)$ and $q|(j-i-1)$. Since $\gcd(p,q)=2$ and $0\leq i<j<n=\frac{1}{2}pq$, we have $j-i-1=0$ that is $j=i+1$. Therefore, $\widehat{\delta}(s_i,a^{n-i+1})=\langle s^1_{n+1\ {\it mod}\ p}, s^2_{n+1\ {\it mod}\ q}\rangle=\langle s^1_{1}, s^2_1\rangle\in F$ and  $\widehat{\delta}(s_j,a^{n-i+1})=\widehat{\delta}(s_{i+1},a^{n-i+1})=\langle s^1_{n+2\ {\it mod}\ p}, s^2_{n+2\ {\it mod}\ q}\rangle=\langle s^1_{2}, s^2_2\rangle\not\in F$. Hence, $s_i$ and $s_j$ are distinguishable.

         \item      $j-i\not\equiv 1\ {\it mod}\ p$. We have  $\widehat{\delta}(s_j,a^{n-j+1})=\langle s^1_{n+1\ {\it mod}\ p}, s^2_{n+1\ {\it mod}\ q}\rangle=\langle s^1_{1}, s^2_1\rangle\in F$ and $\widehat{\delta}(s_i,a^{n-j+1})=\langle s^1_{n-j+1+i\ {\it mod}\ p}, s^2_{n-j+1+i\ {\it mod}\ q}\rangle=\langle s^1_{-j+1+i\ {\it mod}\ p}, s^2_{-j+1+i\ {\it mod}\ q}\rangle$. Since $j-i\not\equiv 1\ {\it mod}\ p$, we have $-j+1+i\not\equiv 0\ {\it mod}\ p$.
              Since $i\not\equiv j \ {\it mod}\ q$, we have $(-j+1+i)\not\equiv 1\ {\it mod}\ q$. Therefore, $\widehat{\delta}(s_i,a^{n-j+1})\not\in F$. We have again that $s_i$ and $s_j$ are distinguishable.
          \end{enumerate}
  \end{enumerate}
We  have therefore shown that the DFA ${\cal A}'$ is minimal and  has $\frac{1}{2}pq$ states. Let us assume that  there is a pvDFA ${\cal B}$ with less than $\frac{1}{2}pq$ states recognizing  the promise problem $A$. We can then get a DFA with less than $\frac{1}{2}pq$ states  recognizing the language $A_{yes}\cup A_{no}$. This would implies  that the DFA ${\cal A}'$ is not minimal.
A contradiction.

 Obviously, $s(A_{yes})=p$ and $s(A_{no})=q$. Therefore, we have proved that $sr(A)= \frac{1}{2}pq=\frac{1}{2}s(A_{yes})s(A_{no})$.

\end{proof}

\end{document}